\documentclass[journal,10pt]{IEEEtran}
\usepackage{amsthm}
\usepackage{dsfont}
\usepackage{amsfonts}
\usepackage{amssymb}
\usepackage{mathrsfs}
\usepackage{mathtools}
\usepackage[english]{babel}
\usepackage{float}
\usepackage[pdftex]{graphicx}
\usepackage{epstopdf}
\usepackage{amsmath}
\usepackage{color}
\usepackage{multirow}
\usepackage{graphicx}
\usepackage[switch]{lineno}
\usepackage[named]{algo}
\usepackage{algpseudocode}
\usepackage{algorithm}
\usepackage{algorithmicx}
\usepackage{bm}
\usepackage{stfloats}
\usepackage{lipsum}
\usepackage[inline]{enumitem}
\usepackage{microtype}
\usepackage{cuted}
\usepackage[caption=false,font=footnotesize]{subfig}
\usepackage{cite}
\usepackage{hyperref}
\usepackage{tikz}
\usetikzlibrary{arrows}
\usepackage{siunitx}
\DeclareMathAlphabet\mathbfcal{OMS}{cmsy}{b}{n}
\newtheorem{theo}{Theorem}
\newtheorem{lem}{Lemma}

\theoremstyle{definition}
\newtheorem{rmk}{Remark}
\newtheorem{Def}{Definition}

\newcommand{\EX}{{\mathbb{E}}}          
\newcommand{\IDM}{{\mathbf{I}}}         
\newcommand{\DIAG}{{\mathrm{diag}}}     
\newcommand{\RE}{{\mathrm{Re}}}         
\newcommand{\IM}{{\mathrm{Im}}}         
\newcommand{\SGN}{{\mathrm{sgn}}}       
\newcommand{\SINC}{{\mathrm{sinc}}}     
\newcommand{\J}{{\mathrm{j}}}
\newcommand{\AM}{\mathbf{A}}

\newcommand{\FM}{\mathbf{F}}
\newcommand{\GM}{\mathbf{G}}

\newcommand{\SM}{\mathbf{S}}
\newcommand{\TM}{\mathbf{T}}
\newcommand{\VM}{\mathbf{V}}
\newcommand{\WM}{\mathbf{W}}

\newcommand{\AV}{\bm{a}}
\newcommand{\BV}{\mathbf{b}}
\newcommand{\CV}{\mathbf{c}}

\newcommand{\GV}{\bm{g}}
\newcommand{\HV}{\mathbf{h}}
\newcommand{\NV}{\bm{n}}
\newcommand{\RV}{\bm{r}}
\newcommand{\QV}{\bm{q}}
\newcommand{\SV}{\mathbf{s}}

\newcommand{\UV}{\bm{u}}

\newcommand{\WV}{\mathbf{w}}
\newcommand{\XV}{\mathbf{x}}
\newcommand{\YV}{\mathbf{y}}
\newcommand{\ZV}{\mathbf{z}}
\newcommand{\SIGMA}{\mathbf{\Sigma}}
\newcommand{\THETA}{\boldsymbol{\theta}}
\newcommand{\PHI}{\boldsymbol{\phi}}
\newcommand{\MU}{\boldsymbol{\mu}}

\newcommand{\LAMBDA}{\boldsymbol{\lambda}}

\newcommand{\ALPHA}{\boldsymbol{\alpha}}
\newcommand{\TAU}{\boldsymbol{\tau}}
\newcommand{\GAMMA}{\boldsymbol{\gamma}}

\newcommand{\EPSILON}{\boldsymbol{\epsilon}}
\newcommand{\ZEROVV}{\mathbf{0}}
\newcommand{\ONEVV}{\mathbf{1}}
\newcommand{\KSET}{{\mathds{K}}}
\newcommand{\SUBJTO}{{\textrm{subject to}}}
\newcommand{\MIN}{{\textrm{minimize}}}
\newcommand*{\Scale}[2][4]{\scalebox{#1}{$#2$}}%
\begin{document}

\title{Localization with One-Bit Passive Radars in Narrowband Internet-of-Things using Multivariate Polynomial Optimization}

\author{\IEEEauthorblockN{Saeid~Sedighi,~\IEEEmembership{Student Member~IEEE,} Kumar Vijay Mishra,~\IEEEmembership{Senior Member~IEEE,} M.~R.~Bhavani~Shankar,~\IEEEmembership{Senior Member~IEEE,} and~Bj\"{o}rn~Ottersten,~\IEEEmembership{Fellow~IEEE}\\}
\thanks{
S. S. acknowledges support via Luxembourg National Research Fund (FNR) under the ACCORDION project
(reference number 11228830). Other authors acknowledge partial support via ERC AGNOSTIC (Grant ID: 742648) and the FNR under the BRIDGES project AWARDS.

The authors are with the Interdisciplinary Centre for Security, Reliability
and Trust (SnT), University of Luxembourg, Luxembourg City L-1855, Luxembourg. E-mail: \{saeid.sedighi@, kumar-mishra@ext, bhavani.shankar@, bjorn.ottersten@\}uni.lu.

The conference precursor of this work was presented at the 2019 IEEE International Workshop on Computational Advances in Multi-Sensor Adaptive Processing (CAMSAP).
}
}

\maketitle

\begin{abstract}
Several Internet-of-Things (IoT) applications provide location-based services, wherein it is critical to obtain accurate position estimates by aggregating information from individual sensors.
In the recently proposed narrowband IoT (NB-IoT) standard, which trades off bandwidth to gain wide coverage, the location estimation is compounded by the low sampling rate receivers and limited-capacity links. We address both of these NB-IoT drawbacks in the framework of passive sensing devices that receive signals from the target-of-interest. We consider the limiting case where each node receiver employs one-bit analog-to-digital-converters and propose a novel low-complexity nodal delay estimation method using constrained-weighted least squares minimization.
To support the low-capacity links to the fusion center (FC), the range estimates obtained at individual sensors are then converted to one-bit data. At the FC, we propose target localization with the aggregated one-bit range vector using both optimal and sub-optimal techniques. The computationally expensive former approach is based on Lasserre's method for multivariate polynomial optimization while the latter employs our less complex iterative joint r\textit{an}ge-\textit{tar}get location \textit{es}timation (ANTARES) algorithm.  Our overall one-bit framework not only complements the low NB-IoT bandwidth but also supports the design goal of inexpensive NB-IoT location sensing. Numerical experiments demonstrate feasibility of the proposed one-bit approach with a $0.6$\% increase in the normalized localization error for the small set of $20$-$60$ nodes over the full-precision case. When the number of nodes is sufficiently large ($>80$), the one-bit methods yield the same performance as the full precision.
\end{abstract}

 
\IEEEpeerreviewmaketitle

\section{Introduction}
\label{sec:intro}
Recent industry estimates project that nearly 75 billion devices will be connected in the Internet-of-Things (IoT) by the year 2025 \cite{ikpehai2018low}. The IoT is envisioned to connect the physical and digital world through extensive instrumentation with sensing, wearable, and intelligent devices \cite{sisinni2018industrial}. A common IoT application is to provide various \textit{localization-based services} \cite{khelifi2019survey,shit2018location}, wherein a large network of devices collects and transmits data to determine the position of entities-of-interest with respect to a node or sensor within the IoT. The location information is critical in order to  gather crucial inference from physical measurements in applications such as military surveillance \cite{kott2016internet}, physiological sensors \cite{rong2019active},  smart homes \cite{dorri2017blockchain}, 
disaster response \cite{han2019harnessing}, and environmental monitoring \cite{mishra2020deep}. 

Global Positioning System (GPS) devices are quite reliable in providing localization measurements in other applications. However, GPS deployment at every IoT node is very expensive in terms of cost and power, especially for networks with massive number of devices. Further, GPS performs poorly in indoor environments. Therefore, many alternative IoT localization methods have been proposed in recent studies \cite{shit2018location}. A promising technology is passive sensor tags that augment existing IoT deployments through backscatter communications \cite{mishra2019multi}. These tags do not have any \textit{active} radio-frequency (RF) chain components thereby leading to huge savings in cost and energy. This is also a practical approach because it is difficult to re-purpose the preset IoT network sensing modalities (usually fixed before the deployment), especially when it comprises millions of devices \cite{perez2016augmenting}. On the other hand, addition of passive sensors does not require changing the deployed IoT hardware or placement of new communications and power sources \cite{ensworth2015every}. 

Since the IoT framework is defined by a massive number of largely battery-powered devices, that also transmit or receive data, the underlying challenges for any communications link in this setting are low power, low data rate, wide coverage, and scalability \cite{xu2017narrowband}. In this context, the 3rd generation partnership project (3GPP) recently introduced narrowband IoT (NB-IoT) system specifications to support wide coverage area, long user lifetime, and low power/cost devices over a narrow bandwidth of 180 kHz \cite{yang2017narrowband}. While not fully backward compatible with existing 3GPP devices, the NB-IoT harmoniously coexists with legacy networks by reusing the functionalities of the latter's design. The reduced NB-IoT bandwidth implies higher transmit power spectral density within the existing 3GPP specifications. This, combined with a soft re-transmission strategy \cite{zhang2019spectrum}, enhances the coverage of NB-IoT over conventional IoT solutions. The ultra-low complexity and low power consumption features of NB-IoT are advantageous for location-based services such as smart parking, smart tracking, and smart home \cite{kellogg2015wi}. In this paper, we focus on passive localization in NB-IoT networks. 

While NB-IoT networks benefit from low bandwidth to enhance their coverage, the same feature imposes challenges in localization by severely limiting the data rate. Commonly used ranging-based localization techniques lose accuracy because of low data rates \cite{gezici2005localization}. In NB-IoT devices, low battery-power is insufficient to handle high sampling rates required to attain necessary localization accuracy \cite{song2017csi,hu2017improving,hu2019time,jeon2018effective}. A popular alternative NB-IoT localization technique is to employ \textit{fingerprinting}, wherein the received signal strength indicator (RSSI) measurements are collected at specified locations during the training phase and then compared with online measurements to determine the location of the target \cite{song2017csi,sallouha2017localization}. However, this approach requires prior knowledge of a detailed RSSI database which may be unavailable or unattainable.
Hence, recent NB-IoT studies explore RSSI-independent signal processing methods such as successive interference cancellation \cite{hu2017improving}, maximum likelihood estimation \cite{hu2019time}, frequency hopping \cite{jeon2018effective} and machine learning \cite{sallouha2019localization}. Our proposed technique is inspired by localization in passive radar \cite{noroozi2015target} not requiring prior RSSI measurements.

The aforementioned works assume that measurements at each node are digitally represented by a large number of bits per sample such that the resulting quantization errors can be neglected. Further, when nodal measurements are sent to a fusion center (FC) for an aggregate decision, full capacity links are assumed. In this paper, contrary to these works, we consider the limiting case wherein the receivers at each node employ one-bit analog-to-digital converters (ADCs), which directly convert node measurements into \emph{complex} data with binary components, each containing one-bit information, by comparing the real and imaginary parts of the node measurements with appropriate thresholds separately and noting the sign. This leads to one-bit per component measurements. Considering the fact that the cost and power consumption of ADCs increase exponentially with the number of quantization bits and sampling frequency \cite{mishra2019toward}, the use of one-bit ADCs supports the low-cost and low-power-consumption features of NB-IoT. We then leverage the recent advances in one-bit signal processing \cite{li2018survey} to estimate the target \textit{range/delay} with respect to a specific node. To cope with the capacity limitations of the nodal links, we assume that, prior to transmission to FC, the receive sensors quantize nodal estimates to one-bit data. The FC then performs target \textit{localization}, i.e. determination of target's position with respect to the entire network, using the one-bit range vector
aggregated from the estimates sent by all the nodes.

Converting analog signals into digital data using a single bit per sample
leads to significant errors in the digital approximation of the original analog signals. This necessitates development of new algorithms for information retrieval from one-bit samples.
One-bit sampling has a rich heritage of research in statistical signal processing \cite{dabeer2008multivariate,Host-Madsen2000,dabeer2006signal} 
and signal reconstruction \cite{cvetkovic2007single}. It was shown in \cite{cvetkovic2007single} that, for band-limited bounded-amplitude square-integrable input signals, a sufficient number of one-bit samples lead to recovery of full-precision data with locally bounded point-wise error, resulting in an exponentially decaying distortion-rate characteristic. In the past few years, one-bit signal processing has received significant attention in numerous modern applications such as array processing \cite{bar2002doa,sedighi2020one}, massive multiple-input multiple-output (MIMO) \cite{Pirzadeh2017}, deep learning \cite{elbir2019joint}, dictionary learning \cite{zayyani2015dictionary}, and radar \cite{zhao2018deceptive}. Most of these works are based on either well-known Bussgang's Theorem \cite{bussgang1952crosscorrelation,bar2002doa,Pirzadeh2017} or compressive sensing techniques \cite{knudson2016one,elbir2019joint,zayyani2015dictionary,zhao2018deceptive}. Further, there are some elegant works on colocated one-bit radar and array processing \cite{ameri2019one,sedighi2020one} which formulate the parameter estimation from one-bit measurements as an optimization problem with linear constraints which can be solved by polynomial-time algorithms. Contrary to previous works on colocated one-bit radar \cite{ameri2019one}, our proposed method investigates widely separated radar setting.

We first formulate the problem of range/time-delay estimation in a clutter-free environment from one-bit samples received by each NB-IoT sensor as a sparse recovery problem. The formulation and approach of the clutter-free scenario is
effectively applicable in a weak clutter environment but the impact of strong clutter is unexamined and left for the future work. We show that, unlike infinite precision sampling, oversampling could improve the range/delay estimation performance in one-bit sampling. Further, oversampling leads our proposed approach to be able to achieve a considerably high resolution for time-delay estimation despite the narrow bandwidth used in NB-IoT. Toward dealing with the capacity limitations of the backhaul links, we assume that each sensor forwards an one-bit conversion of their range measurements to the FC. Collecting these one-bit measurements at the FC, we formulate the passive localization problem using the bistatic range-difference model. Note that the passive localization with NB-IoT sensors has a model similar to that of a passive radar \cite{noroozi2015target}. The passive radar localization has been considered in \cite{noroozi2015target} in the high-resolution ADC framework in which full-precision range measurements are assumed. This usually results in a system of several equations that are solved conventionally by the least squares (LS) method.
In this context, apart from application to NB-IoT localization, 
ours is the first work in the context of one-bit sampling in a passive and distributed radar setting.

In our bistatic range-difference model, recovering locations from one-bit samples requires minimizing a cost function that is a non-negative polynomial in range measurement variables and subjected to polynomial inequalities defined by the positive-valued samples (the one-bit range measurements). The general approach to solving this problem is to re-cast the feasibility of this finite system of polynomial constraints in terms of an equivalent polynomial that involves squares of (unknown) polynomials \cite{shor1987class}. However, it is rather difficult to express a non-negative multivariate polynomial as a sum-of-squares. To address this, we employ Lasserre's general solution approach for polynomial optimization problems via semi-definite programming (SDP) using methods based on moment theory \cite{lasserre2001global}. Our novel formulation jointly estimates the full-precision data as well as the target location. While this method could attain the global minimum, its computational complexity grows considerably with increase in the number of NB-IoT sensors. In order to reduce the computational complexity, we trade accuracy with complexity by proposing a novel sub-optimal iterative joint r\textit{an}ge-\textit{tar}get location \textit{es}timation (ANTARES) algorithm. We also derive the Cram\'{e}r-Rao bound (CRB) for localization with one-bit nodal range measurements and use it as benchmark for assessing the estimation performance of the proposed optimal and sub-optimal algorithms. Numerical results show that when sufficiently large number of NB-IoT nodes are available, the optimal approach yields same performance as the full-precision and ANTARES leads to only $0.43$\% increase in the normalized localization error. Further, the normalized localization error rises minimally by $2.2$\% and $0.6$\% for a smaller set of $20$-$60$ nodes using ANTARES and optimal algorithm, respectively, over the full precision case.

Preliminary results of this work appeared in our conference publication \cite{sedighi2019localization}, where performance analysis was not included and only Lasserre's approach was considered. In this paper, we also investigate the one-bit time-delay estimation for the oversampled scenario and present ANTARES algorithm. In summary, our work provides a robust framework for location-based services in NB-IoT, does not require prior RSSI measurements, performs target delay estimation with one-bit samples, yields localization using limited capacity links, and is computationally efficient. Further, our work also has connections with the recent developments in spectrum sharing and joint radar-communications (JRC) design \cite{mishra2019toward,dokhanchi2019mmwave}. Unlike some recent works \cite{liu2020codesign} where new waveforms are developed for distributed JRC, our work exploits existing NB-IoT signaling for a sensing application.

The rest of the paper is organized as follows. In the next section, we describe the system and signal model of the passive localization problem via the NB-IoT sensors.
We introduce our one-bit nodal range estimation algorithm in Section~\ref{sec:time-delay estimation}. Then, using these estimates, we localize the target at FC in Section~\ref{sec:localization} through a polynomial optimization. We validate our models and methods through numerical experiments in Section~\ref{sec:numexp} before concluding in Section~\ref{sec:summ}.

Throughout this paper, we refer the vectors and matrices by lower- and upper-case bold-face letters, respectively. The superscripts $(\cdot)^T$ and $(\cdot)^H$ indicate the transpose and Hermitian (conjugate transpose) operations, respectively. $[\AM]_{i,j}$ and $[\AV]_i$ indicate the $(i,j)$-th and $i$-th entry of $\AM$ and $\AV$, respectively. The notations $\|\AV\|_1$ and $\|\AV\|_2$ stand for $\ell_1$-norm and $\ell_2$-norm of the vector $\AV$, respectively; $|a|$ and $\lceil a \rceil$ represent the absolute value of and the least integer greater than or equal to the scalar $a$, respectively; the estimates of $\AV$ and $a$ are indicated by $\widehat{\AV}$ and $\widehat{a}$, respectively; superscript within parentheses as $(\cdot)^{(k)}$ indicates the value at $k$-th iteration; a diagonal matrix with the diagonal vector $\AV$ is $\DIAG(\AV)$; the real and imaginary parts of the complex number $a$ are $\RE\{a\}$ and $\IM\{a\}$, respectively; $\deg(.)$ is the degree of a polynomial; $\EX\{.\}$ stands for the statistical expectation; $\mathbf{I}_M$ is the $M \times M$ identity matrix; $\AM^{\dagger}$, $\Pi_{\AM}=\AM\AM^{\dagger}$ and $\Pi^{\bot}_{\AM} = \IDM_M - \AM\AM^{\dagger}$ indicate the pseudo-inverse, the projection matrix onto the range space and the projection matrix onto the null space of the full column rank matrix $\AM$, respectively; ${\cal R}(\mathbf{A})$ and ${\cal N}(\mathbf{A})$ represent the range and null spaces  of $\mathbf{A}$, respectively; $\AM \succeq \ZEROVV$ and $\AV \succeq \ZEROVV$ indicate a positive semidefinite matrix and a vector with all elements greater than or equal to zero, respectively. The symbol $\odot$ represents the Hadamard (element-wise) product and $\SGN(\cdot)$ stands for the sign function. The notation $\frac{\partial f}{\partial_{x}}$ is the partial derivative of the function $f$ with respect to the variable $x$.

\section{System Model}
\label{sec:model}
Consider a source, say, a communications base-station whose location in Cartesian coordinates is $\begin{bmatrix} \delta^x_b & \delta^y_b & \delta^z_b\end{bmatrix}^T \in \mathds{R}^{3 \times 1}$. The source transmits a known baseband single-tone NB-IoT signal $s(t) \in \mathds{C}$ with bandwidth $B$. 
As per NB-IoT specifications, the signal has spectrum limited to 180 kHz. It is similar to LTE with fewer (1, 3, 6, or 12) subcarriers with normal cyclic prefix \cite{xu2017narrowband,loulou2020multiplierless} and employs rotated phase shift keying (PSK) constellations, either $\pi/2$ binary PSK ($\pi/2$-BPSK) or $\pi/4$ quadrature PSK ($\pi/4$-QPSK). 
The resulting signal is\par\noindent\small
\begin{align}
s(t)=\sum_{k=0}^{N_c-1} a_{k} e^{\mathrm{j} k \frac{\pi}{M}} g(t-k T_c),\;\;\;0\le t < T,
\end{align}\normalsize
where $a_{k} \in \{\pm 1\}$ for $\pi/2$-BPSK and $a_{k} \in \{\pm 1, \pm \mathrm{j} \}$ for $\pi / 4$-QPSK are known pilot symbols, $M$ is the alphabet size (2 for $\pi/2$-BPSK and 4 for $\pi/4$-QPSK), $N_c$ is the maximum number of symbols allowed during the transmission, $T$ denotes the observation interval, $T_c$ is the symbol period, and $g(t)$ is the pulse shaping filter impulse response with bandwidth $B$.
\begin{figure}[t]
\centering
\includegraphics[width=1.0\columnwidth]{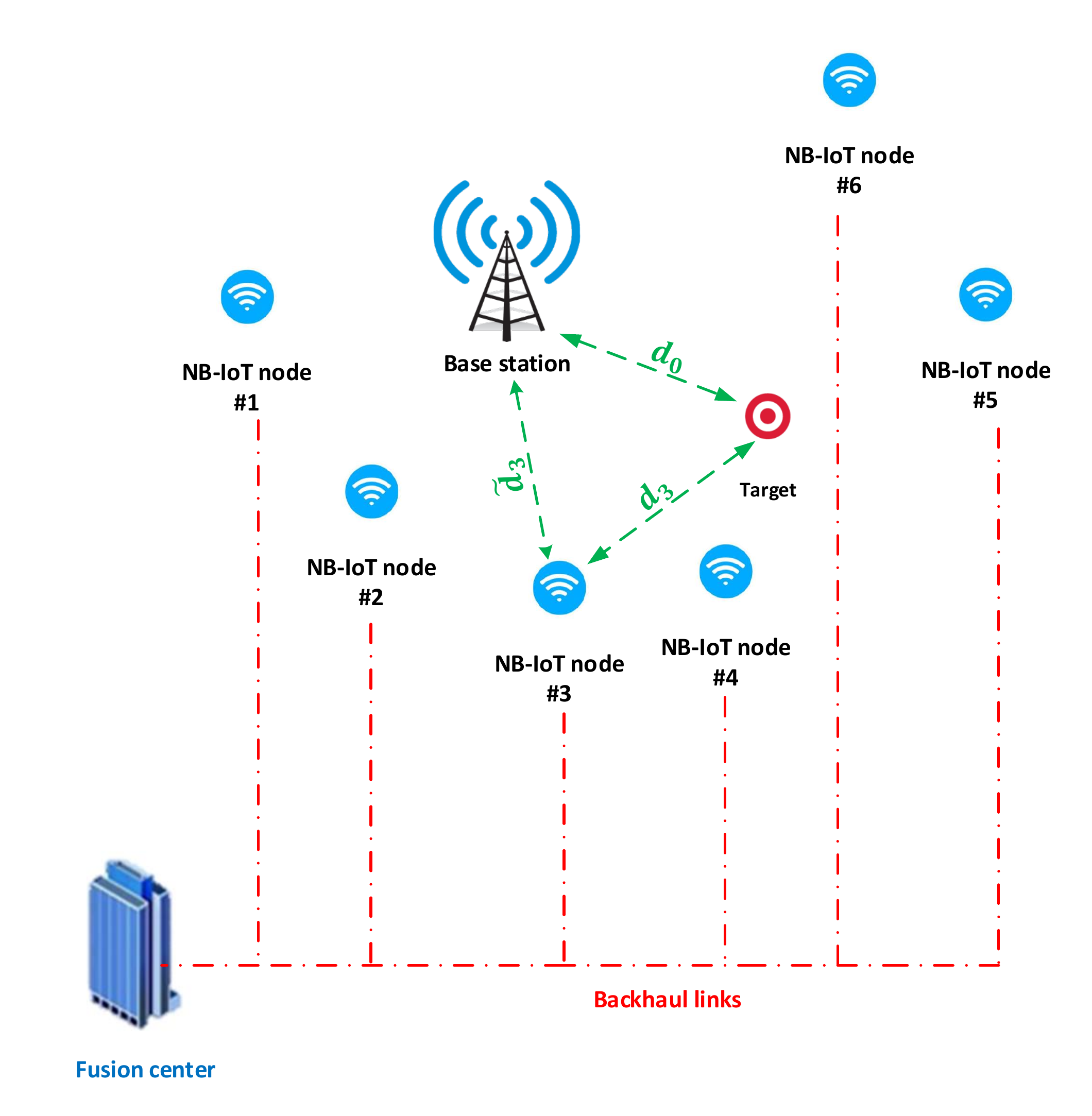}
\DeclareGraphicsExtensions.
\vspace*{-6mm}
\caption{Illustration of the localization scenario. The NB-IoT \#$1$, \#$2$, $\cdots$, nodes (blue) are passive sensors (located at distances $\tilde{d}_1,\tilde{d}_2, \cdots, \tilde{d}_6$ from the base station). The nodes receive the signal from the source bounced off from a target-of-interest (red) located at distances $d_1, d_1, \cdots, d_6$ from the nodes and $d_0$ from the base station. In our proposed model, the nodes employ one-bit ADCs to sample the received signal and estimate the range. The estimated range at each node is quantized and then forwarded to the FC for an aggregated estimate.}
\label{fig-1}
\end{figure}

The transmit signal is bounced off from the target-of-interest located at $\begin{bmatrix} \delta^x & \delta^y & \delta^z\end{bmatrix}^T \in \mathds{R}^{3 \times 1}$. In a typical NB-IoT setting, a target could be a subject carrying a mobile phone, an intelligent vehicle or a robot. The backscattered signal is then received by $M$ distinct NB-IoT sensor nodes. The location of the $m$-th node is $\begin{bmatrix} \delta^x_m & \delta^y_m & \delta^z_m \end{bmatrix}^T \in \mathds{R}^{3 \times 1}, m \in \mathds{M} \doteq \{1, 2, \cdots, M\}$.  These nodes are synchronized with the base-station (Fig.~\ref{fig-1}). Synchronization could be provided by sending a periodic synchronization signal from the base-station to the NB-IoTs, including timing information of the base-station, while the base-station maintains a constant clock using either receiving a reference time from GPS or an atomic clock. After receiving the the base-station timing information, NB-IoTs are able to accurately synchronize their clocks with the base-station clock \cite{Simeone2008,Wu2011,Ali2017,Zhang2017}. More detailed information about the periodic synchronization signal and the synchronization mechanism in NB-IoT systems are provided in \cite{Ali2017,Zhang2017}, and the references therein. Synchronization may be also achieved through the use of protocols such as IEEE 1588 generic precision time protocol (gPTP) \cite{ieee2008precision}, network time protocol (NTP) \cite{mills1991internet} and wireless PTP \cite{garg2017wireless}. These cost-effective clock synchronization protocols are also popular in other applications, including electrical grid networks, cellular base-station synchronization, industrial control, and vehicular systems \cite{karthik2020robust,wang2020displaced}.

If the distance between the source and the target is $d_0$ and that between the target and the $m$-th NB-IoT node is \par\noindent\small
\begin{align}
\label{eq-2}
d_m\!=\!\sqrt{(\delta^x_m-\delta^x)^2+(\delta^y_m-\delta^y)^2+(\delta^z_m-\delta^z)^2},\;1 \!\leq\! m \!\leq\! M,
\end{align}\normalsize
then the true target range with respect to the $m$-th NB-IoT node is
\par\noindent\small
\begin{align}
\label{eq-3}
r_m = d_m+ d_0,\;1 \!\leq\! m \!\leq\! M.
\end{align}\normalsize

The propagation is non-dispersive and the base-station signal received by the NB-IoT nodes includes a direct line-of-sight (LoS) path from the base-station to the nodes and an indirect non-LoS (NLoS) path from the base-station to the target and then to the nodes. The demodulated baseband analog signal received at $m$-th sensor is
\par\noindent\small
\begin{align}
\label{eq-analog-sig-recieve}
\breve{y}_m(t) = \widetilde{\alpha}_m s(t - \widetilde{\tau}_m) + \alpha_m s(t - \tau_m) + \breve{n}_m(t),
\end{align}\normalsize
where $\widetilde{\alpha}_m \in \mathds{C}$ ($\alpha_m \in \mathds{C}$) and $\widetilde{\tau}_m \in \mathds{R}$ ($\tau_m \in \mathds{R}$) are the attenuation coefficient and time-delay of the propagation channel for the direct (indirect) path, respectively; and $\breve{n}_m(t) \in \mathds{C}$ denotes additive white noise following a circular-symmetric complex Gaussian distribution with variance $N_m > 0$. The unknown time delay $\tau_m$ is linearly proportional to $r_m$, i.e. $\tau_m = r_m/c$ where $c=3\times 10^8$ m/s is the speed of light. The unknown direct path delay $\widetilde{\tau}_m$ is also linearly proportional to the distance between the $m$-th node and the base station. i.e., $\widetilde{\tau}_m = \widetilde{d}_m/c$ where $\widetilde{d}_m = \sqrt{(\delta^x_m-\delta^x_b)^2+(\delta^y_m-\delta^y_b)^2+(\delta^z_m-\delta^z_b)^2}$ denotes the distance between the $m$-th node and the base station.

The baseband signal is filtered by an ideal low-pass filter with bandwidth $B$ and frequency response
\par\noindent\small
\begin{align}
H(\Omega)= \left\{\begin{array}{ll} 1, & |\Omega| \leq 2\pi B, \\
0, & {\rm otherwise}.
\end{array}\right.
\end{align}\normalsize
This low-pass filtering of the signal $\breve{y}_m(t)$ yields
\par\noindent\small
\begin{align}
\label{eq-filter-sig-recieve}
y_m(t) = \widetilde{\alpha}_m s(t - \widetilde{\tau}_m) + \alpha_m s(t - \tau_m) + n_m(t),
\end{align}\normalsize
where $n_m(t)$ is the filtered noise trail whose auto-correlation is\par\noindent\small
\begin{align}
\label{eq-noise-autocorr}
 R_{n_m}(t_1\!-\!t_2) &\!=\! \frac{1}{2\pi} \int_{- \infty}^{\infty} N_m  | H(\Omega) |^2 e^{-\J \Omega (t_1-t_2)} {\rm d}\Omega \nonumber\\
 &\!=\! 2BN_m \SINC (2B(t_1\!-\!t_2)),
\end{align}\normalsize
where $\SINC(u) = \frac{\sin (\pi u)}{\pi u}$.

%
\begin{figure}[t]
\centering
\begin{tikzpicture}[scale=0.4]
\draw[-,dashed] (-7,-3.2) -- (-7,4);
\draw[-,dashed] (-2,-3.2) -- (-2,4);
\draw[-,dashed] (-7,-3.2) -- (-2,-3.2);
\draw[-,dashed] (-7,4) -- (-2,4);
\draw[->,semithick] (-10,2) -- (-6,2);
\draw [semithick] (-6,1) rectangle (-3,3);
\draw[->,semithick] (-4.5,-1) -- (-4.5,1);
\draw[->,semithick] (-3,2) -- (1,2);
\draw [semithick] (1,1) rectangle (4,3);
\draw[->,semithick] (4,2) -- (7,2);
\node[] at (-8.5,3){$y_m(t)$};
\node[] at (-0.5,3){$\YV_m$};
\node[] at (5.5,3){$\ZV_m$}; 
\node[] at (-4.5,2){CDC}; 
\node[] at (2.5,2){$Q(.)$}; 
\node[] at (-4.5,-2){$T_s = \frac{1}{2 \vartheta B}$};
\end{tikzpicture}
\caption{Conceptual representation of the oversampled one-bit ADC. The CDC block represents the digitizer operating at sampling rate of $1/T_s$. A quantizer $Q(\cdot)$ then converts the digital samples into a one-bit data stream.}
\label{fig-2}
\end{figure}
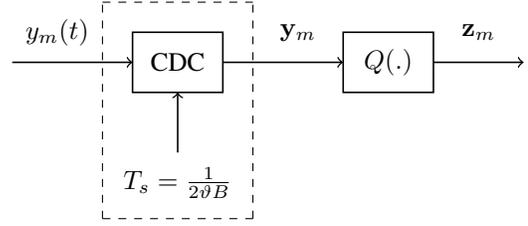
Each NB-IoT node is equipped with a one-bit ADC which admits binary samples of the corresponding $y_m(t)$  during the observation interval $[0, T)$. The ADC sampling frequency $f_s = \frac{1}{T_s} = 2\vartheta B$, where $\vartheta$ is an integer greater than or equal to one, referred to as the oversampling factor.
 Figure~\ref{fig-2} conceptually depicts a one-bit ADC which comprises a Continuous-to-Discrete Converter (CDC) with sampling frequency $f_s \!=\! 2 \vartheta B$ followed by a one-bit quantizer. The CDC produces $L \!=\! \frac{T}{T_s}=2 \vartheta BT$ discrete samples of $y_m(t)$ during the time interval $[0, T)$. Stacking all discrete samples produces a $\mathds{C}^{L \times 1}$ vector
\par\noindent\small
\begin{align}
\label{eq-samp-sig}
    \mathbf{y}_m = \widetilde{\alpha}_m \SV(\widetilde{\tau}_m) + \alpha_m \SV(\tau_m) + \NV_m,
\end{align}\normalsize
where $[\mathbf{y}_m]_l = y_m((l-1)T_s)$, $[\SV(\widetilde{\tau}_m)]_l = s((l-1)T_s - \widetilde{\tau}_m)$, $[\SV(\tau_m)]_l = s((l-1)T_s - \tau_m)$, and $[\NV_m]_l = n_m((l-1)T_s)$ for $l = 1, 2, \cdots, L$. From \eqref{eq-noise-autocorr} and Gaussianity of $n_m(t)$, vector $\NV_m$ follows a zero-mean complex Gaussian distribution with the covariance
\par\noindent\small
\begin{align}
\EX\{\NV_m \NV_m^H\}  =  \sigma_m^2 \SIGMA ~~\in \mathds{C}^{L \times L}
\end{align}\normalsize
where
$[\SIGMA]_{i,j} = \SINC\left(\frac{|i-j|}{\vartheta}\right)$ and
$\sigma_m^2 = 2BN_m$.

The quantizer, represented by a function $Q(\cdot)$, converts the discrete samples into binary data by comparing each sample to a known threshold and then measuring the sign of the real and imaginary parts of the resulting difference. These one-bit measurements at the $m$-th NB-IoT node are
\begin{align}
\label{eq-quantizer}
    \ZV_m =
     Q(\YV_m),
\end{align}
where the $l$-th element of $Q(\YV_m)$ is \par\noindent\small
\begin{align}
\label{eq-onebit}
&[Q(\YV_m)]_l \\&= \frac{1}{\sqrt{2}} \SGN(\RE\{[\YV_m]_l - [\GAMMA_m]_l\}) + \frac{\J}{\sqrt{2}}~ \SGN(\IM\{[\YV_m]_l - [\GAMMA_m]_l\}). \nonumber
\end{align}\normalsize
with $\GAMMA_m \in \mathds{C}^{L \times 1}$ are known thresholds levels.

The nodal processing at each NB-IoT receiver entails estimation of the target time-delays, and hence the range, from one-bit samples $\ZV_m$. In the next section, we devise a method for one-bit time-delay estimation.

\section{Time-Delay Estimation with One-Bit Samples}
\label{sec:time-delay estimation}
Several approaches have been proposed in the literature to estimate range (time-delay) of targets from one-bit samples with most formulating this as an optimization problem. For example, the covariance matrix formulation of \cite{ameri2019one} employs cyclic optimization method to extract the range along with other parameters. Other recent works using only one sensor exploit sparsity of the target scenario to estimate unknown parameters by applying techniques such as $\ell_1$-norm minimization \cite{zahabi2020one} and log-relaxation
 \cite{zhu2020target} to solve the resulting optimization. In our passive NB-IoT sensor set-up, the objective function is a variation of weighted least squares (WLS) that we minimize via $\ell_1$-norm regularization to estimate $\tau_m$ using the one-bit quantized observations, i.e., $\ZV_m$.
In conventional passive radars, direct and indirect path signals are recorded in separate reference and surveillance channels, respectively. However, the direct signal may seep into the surveillance channel and mask the relatively weaker indirect signal. In such cases, adaptive filters are employed to first suppress the direct signal in the surveillance channel \cite{garry2017evaluation}. However, our NB-IoT scenario is an opportunistic sensing application where the receivers are not equipped to record separate channels. Moreover, as explained next, the (additive) overlap of direct signal with the target echo is useful because the former is used to estimate the latter in our formulation. Here, we also remark that there are passive radar applications where direct signal suppression is not crucial. For example, this requirement is often relaxed in passive sensing using communications satellites because of the relatively weak power of the direct path satellite signal than, say, commonly used broadcasting signals \cite{daniel2017design}.

\vspace{-2mm}
\subsection{Constrained-Weighted Least Squares Minimization}
\label{sec:time-delay algorithm}
Equation \eqref{eq-samp-sig} can be transformed to the frequency domain by multiplying both sides by an $L \times L$ Discrete Fourier Transform (DFT) matrix $\FM$, whose $(n,k)$-th entry is $e^{\frac{-\mathrm{j}2\pi nk}{L}}$. This yields
\par\noindent\small
\begin{align}
\label{eq-dft}
\FM \YV_m = \widetilde{\alpha}_m \DIAG(\overline{\SV}_{\widetilde{\tau}_m})\AV(\widetilde{\tau}_m) + \alpha_m \DIAG(\overline{\SV}_{\tau_m}) \AV(\tau_m)  + \overline{\NV}_m,
\end{align}\normalsize
where $\overline{\NV}_m = \FM \NV_m$, $[\AV(u)]_l = e^{-\J 2\pi \frac{(l-1)u}{LT_s}}$ for $0 \leq l \leq L-1$ and $\overline{\SV}_{u} = \FM \SV_{u} $ with
\par\noindent\small
\begin{align}
[\SV_u]_l = \left\{\begin{array}{ll} s((l-1)T_s) & 1 \leq l \leq  L - \lfloor \frac{u}{L} \rfloor,\\
0 & {\rm otherwise}.
\end{array}\right.
\end{align}\normalsize
Let us discretize the continuous space of the time delay, i.e., $[0, T)$, into a given set of $N \geq L$ grid points, i.e., $\{\overline{\tau}_{m,1}, \cdots, \overline{\tau}_{m,N}\}$ \cite{yang2018sparse}. This discretization transforms \eqref{eq-dft} into the following sparse model 
\par\noindent\small
\begin{align}
\label{eq-dft-re2}
\FM \YV_m =  [\overline{\SM} \odot \AM(\overline{\TAU}_m)] \overline{\boldsymbol{\alpha}}_m +  \overline{\NV}_m
\end{align}\normalsize
where $\AM(\overline{\TAU}_m) = \begin{bmatrix} \AV(\overline{\tau}_{m,1}) & \cdots &  \AV(\overline{\tau}_{m,N}) \end{bmatrix} \in \mathds{C}^{L \times N}$, $\overline{\SM}  = \begin{bmatrix} \overline{\SV}_{\overline{\tau}_{m,1}}  & \cdots &  \overline{\SV}_{\overline{\tau}_{mN1}} \end{bmatrix} \in \mathds{C}^{L \times N}$ and $\overline{\boldsymbol{\alpha}}_m = \begin{bmatrix} \overline{\alpha}_{m,1} & \cdots & \overline{\alpha}_{m,N} \end{bmatrix} \in \mathds{C}^{N \times 1}$ is a sparse vector with
\par\noindent\small
\begin{align}
\label{alpha-sparse}
  [\overline{\boldsymbol{\alpha}}_m]_k = \left\{\begin{array}{ll} \alpha_m, & {\rm if}~ \overline{\tau}_{m,k} = \tau_m,\\
  \widetilde{\alpha}_m, & {\rm if}~ \overline{\tau}_{m,k} = \widetilde{\tau}_m,\\
0, & {\rm otherwise}.
\end{array}\right.
\end{align}\normalsize
The waveform $\SV$ is known at NB-IoT receiver. Hence, the problem is to find $\YV_m$ and a sparse vector $\overline{\boldsymbol{\alpha}}_m$ which are consistent with the model in \eqref{alpha-sparse} as well as one-bit measurments $\ZV_m$. In consequence, the time-delay estimation problem can be formulated as follows \cite{zahabi2020one}
\par\noindent\small
\begin{align}
\label{eq-optimization}
\hspace{-4mm}
\begin{array}{ll}
\underset{\YV_m, \overline{\boldsymbol{\alpha}}_m}{\textrm{minimize}} & \|\overline{\boldsymbol{\alpha}}_m\|_1 + \rho \| \WM \left[\FM \YV_m -  [\overline{\SM} \odot \AM(\overline{\TAU}_m)] \overline{\boldsymbol{\alpha}}_m \right] \|^2_2 \\
\SUBJTO & \RE\{\ZV_m\} \odot \RE\{\YV_m - \GAMMA_m\} \succeq \ZEROVV, \\
& \IM\{\ZV_m\} \odot \IM\{\YV_m - \GAMMA_m\} \succeq \ZEROVV.
\end{array}    
\end{align}\normalsize
where $\rho$ is a regularization parameter and $\WM = \SIGMA^{-\frac{1}{2}} \FM^H$ is a weighting matrix. The first term in the objective of \eqref{eq-optimization} promotes sparsity in $\overline{\boldsymbol{\alpha}}_m$ while the second term is a
WLS criterion that penalizes the model mismatch in \eqref{eq-dft-re2} considering the fact that the additive noise in \eqref{eq-dft-re2} follows a circular-symmetric complex Gaussian distribution with the covariance matrix $\sigma_m \FM \SIGMA \FM^H$. Further, linear constraints arise because one-bit quantized and discrete samples must share the same sign. Introducing a slack variable $\XV_m = \SIGMA^{-\frac{1}{2}} \FM^H \left[\FM \YV_m -  [\overline{\SM} \odot \AM(\overline{\TAU}_m)] \overline{\boldsymbol{\alpha}}_m \right]$, \eqref{eq-optimization} becomes
\par\noindent\footnotesize
\begin{align}\label{eq-op-re-con}
\hspace{-2mm}
\begin{array}{ll}
\underset{\XV_m, \overline{\boldsymbol{\alpha}}_m}{\textrm{minimize}} & \|\overline{\boldsymbol{\alpha}}_m\|_1 + \rho\| \XV_m \|^2_2 \\
\SUBJTO & \RE\{\ZV_m\} \!\odot\! \RE\{ \FM^H [\overline{\SM} \odot \AM(\overline{\TAU}_m)] \overline{\boldsymbol{\alpha}}_m \!+\!    \SIGMA^{\frac{1}{2}} \XV_m \!-\! \GAMMA_m\} \succeq \ZEROVV, \\
& \IM\{\ZV_m\} \!\odot\! \IM\{ \FM^H [\overline{\SM} \odot \AM(\overline{\TAU}_m)] \overline{\boldsymbol{\alpha}}_m \!+\! \SIGMA^{\frac{1}{2}} \XV_m \!-\! \GAMMA_m\} \succeq \ZEROVV.
\end{array}
\end{align}\normalsize
The above problem comprises minimization of a convex objective function with linear constraints and can be solved efficiently \cite{boyd2004convex}.

The solution of \eqref{eq-op-re-con} yields estimate of $\overline{\boldsymbol{\alpha}}_m$ which has two non-zero elements at indices $k_1$ and $k_2$. From this, we find $\widehat{\TAU}_m = \left[\frac{(k_1-1)T}{N},\frac{(k_2-1)T)}{N}\right]^T$. The estimated unknown time delay corresponding to the indirect path is then \par\noindent\small
\begin{align}
\widehat{\tau}_m = \max \{[\widehat{\TAU}_m]_1, [\widehat{\TAU}_m]_2\},
\end{align}\normalsize
\vspace{-5mm}
\begin{lem}
\label{lem-1-2}
$\widehat{\tau}_m$ is a consistent estimate of $\tau_m$.
\end{lem}
\begin{proof}
See Appendix \ref{App-A-A}.
\end{proof}
Hence, the a consistent estimate of the range of the target is given by
$\widehat{r}_m=c \widehat{\tau}_m$.
%
%
\vspace{-2mm}
\subsection{Improved Performance with Oversampling
}
\label{sec:oversampling}
It is possible to improve the recovery performance if the one-bit ADCs sample at a rate higher than the Nyquist. Note that the samples are still quantized to only single bits. In this section, we analyze the effect of oversampling.

In case of oversampling, let replace the CDC module in Fig.~\ref{fig-2} with an equivalent system (Fig.~\ref{fig-3}) composed of a CDC that samples $y_m(t)$ at the Nyquist rate followed by an $\vartheta$-fold upsampling. A low-pass filter with frequency response
\par\noindent\small
\begin{align}
\widetilde{H}(e^{\J \Omega})= \left\{\begin{array}{ll} \vartheta, & |\Omega| \leq \frac{\pi}{\vartheta}, \\
0, & {\rm Otherwise},
\end{array}\right.
\end{align}\normalsize
outputs the oversampled data $\mathbf{y}_m$.
The oversampled $\YV_m$ and Nyquist-sampled $\widetilde{\YV}_m$ (see Fig.~\ref{fig-3}) are related as \cite{oppenheim2009discrete}
\begin{figure}[t]
\centering
\begin{tikzpicture}[scale=0.4]
\draw[-,dashed] (-6,-3.2) -- (-6,4);
\draw[-,dashed] (9,-3.2) -- (9,4);
\draw[-,dashed] (-6,-3.2) -- (9,-3.2);
\draw[-,dashed] (-6,4) -- (9,4);
\draw[->,semithick] (-9,2) -- (-5,2);
\draw [semithick] (-5,1) rectangle (-2,3);
\draw[->,semithick] (-3.5,-1) -- (-3.5,1);
\draw[->,semithick] (-2,2) -- (1,2);
\draw [semithick] (1,1) rectangle (4,3);
\draw [->,thick] (2.2,1.6) -- (2.2,2.4);
\draw[->,semithick] (4,2) -- (5,2);
\draw [semithick] (5,1) rectangle (8,3);
\draw[->,semithick] (8,2) -- (12,2);
\node[] at (-7.5,3){$y_m(t)$};
\node[] at (-0.5,3){$\widetilde{{\YV}}_m$};
\node[] at (10.5,3){$\YV_m$};
\node[] at (-3.5,2){CDC}; 
\node[] at (2.8,2){$\vartheta$};
\node[] at (6.5,2){$\widetilde{H}(e^{\J \Omega})$}; 
\node[] at (-3.5,-2){$T_s^{'} = \frac{1}{2 B}$}; 
\end{tikzpicture}
\caption{An equivalent representation of Fig.~\ref{fig-3} to show both oversampled $\YV_m$ and Nyquist-sampled $\widetilde{\YV}_m$.}
\label{fig-3}
\end{figure}
\par\noindent\footnotesize
\begin{align}
\label{eq-linear-re}
&[\YV_m]_l = \sum_{p=1}^{L/\vartheta} [\widetilde{\YV}_m]_p \SINC \left( \frac{l-1}{ \vartheta} - p+1 \right)\\
&=\begin{dcases}
           [\widetilde{\YV}_m]_p, \;\;\;\;\;\;\;\;\;\;\;\;\;\;\;\;\;\;\;\;\;\;\;\;\;\;\;\;\;\;\;\;\;\;\;\;\; \text{if}~ l = (p-1) \vartheta+1, \; 1 \leq p \leq L/\vartheta , \\
            \sum\limits_{p=1}^{L/\vartheta} [\widetilde{\YV}_m]_p \SINC \left( \frac{l-1}{\vartheta} \!-\! p \!+\! 1 \right), \; \text{otherwise}. 
          \end{dcases} \nonumber
\end{align}\normalsize
Indeed, \eqref{eq-linear-re} implies that $\frac{L}{\vartheta}$ elements of $\YV_m$ are exactly equal to those of $\widetilde{\YV}_m$; and the other elements of $\YV_m$ are obtained from linear combinations of the elements of $\widetilde{\YV}_m$. Let $[\overline{\YV}_m]_l = [\YV_m]_l$ for $l \neq (p-1) \vartheta+1$ and $1 \leq p \leq L/\vartheta$ and $\mathbfcal{I}(.|\THETA)$ denote the Fisher Information Matrix (FIM) with
respect to the parameter vector $\THETA$. The linear dependence of $\overline{\YV}_m$ and $\widetilde{\YV}_m$ implies that $\mathbfcal{I}(\overline{\YV}_m|\widetilde{\YV}_m,\TAU_m, \ALPHA_m)= \ZEROVV$. Hence, it follows from the chain
rule of FIM \cite{Zamir1998} that
\begin{align}
\mathbfcal{I}(\YV_m|\TAU_m, \ALPHA_m) =  \mathbfcal{I}(\widetilde{\YV}_m|\TAU_m, \ALPHA_m).
\end{align}
This means that oversampling has no impact on the accuracy of the time-delay estimation using full-precision data in our model.

Now let us consider the effect of oversampling on the accuracy of the time-delay estimation using one-bit data. Substituting \eqref{eq-linear-re} into \eqref{eq-quantizer} yields
\par\noindent\footnotesize
\begin{align}
\label{eq-nonlinear}
&[\ZV_m]_l = Q([\YV_m]_l)\\
&\hspace{-2mm}=\left\{\begin{array}{ll}
           [\widetilde{\ZV}_m]_p, \;\;\;\;\;\;\;\;\;\;\;\;\;\;\;\;\;\;\;\;\;\;\;\;\;\;\;\;\;\;\;\;\;\;\;\;\;\;\;\;\;\;\;\;\;\text{if}~ l \!=\! (p-1) \vartheta+1, \; 1 \!\leq\! p \!\leq\! L , \\
            Q\left(\sum\limits_{p=1}^{L} [\widetilde{\YV}_m]_p \SINC \left( \frac{l-1}{ \vartheta} - p+1 \right)\right),\; \text{otherwise}, 
          \end{array}
\right.      \nonumber
\end{align}\normalsize
where $\widetilde{\ZV}_m = Q(\widetilde{\YV}_m)$ contains the one-bit data at the Nyquist rate. From \eqref{eq-nonlinear}, we deduce that whereas $\frac{L}{\vartheta}$ elements of $\ZV_m$ are exactly equal to those of $\widetilde{\ZV}_m$, the  remaining elements of $\ZV_m$, denoted by  $\overline{\ZV}_m \in \mathds{C}^{(1-\frac{1}{\vartheta})L \times 1}$, can not be constructed from linear combinations of the elements of $\widetilde{\ZV}_m$ like the full-precision case. In other words, \eqref{eq-nonlinear} indicates that while $\widetilde{\ZV}_m$ provides information about only the signs of $\widetilde{\YV}_m$, $\overline{\ZV}_m$ provides additional information on the signs of the linear combinations of $\widetilde{\YV}_m$. Therefore, in general, $\mathbfcal{I}(\overline{\ZV}_m|\widetilde{\ZV}_m,\TAU_m, \ALPHA_m) \succeq \ZEROVV$. From the chain
rule of FIM \cite{Zamir1998}, we have
\par\noindent\small
\begin{align}
\label{eq-FIM}
\mathbfcal{I}(\ZV_m|\TAU_m, \ALPHA_m) \!=\!  \mathbfcal{I}(\widetilde{\ZV}_m|\TAU_m, \ALPHA_m) \!+\! \mathbfcal{I}(\overline{\ZV}_m|\widetilde{\ZV}_m,\TAU_m, \ALPHA_m).
\end{align}
Considering \eqref{eq-FIM} and $\mathbfcal{I}(\overline{\ZV}_m|\widetilde{\ZV}_m,\TAU_m, \ALPHA_m) \succeq \ZEROVV$, we observe
\begin{align}
\label{eq-FIM2}
\mathbfcal{I}(\ZV_m|\TAU_m, \ALPHA_m) \succeq  \mathbfcal{I}(\widetilde{\ZV}_m|\TAU_m, \ALPHA_m)
  \end{align}
This implies that oversampling could enhance the parameter estimation performance when one-bit quantized data is used.
 \section{Target Localization with One-Bit Samples}
\label{sec:localization}
In order to comply with bandwidth and power limitations, each of the $M$ sensors converts its nodal range measurements into a binary sample $w_m$ by comparing it to a positive threshold $\lambda_m > 0$, i.e.,
\begin{align}
\label{Eq-3}
w_m=\SGN(r_m - \lambda_m).
\end{align}
All nodes forward this binary range and the corresponding thresholds to the FC which localizes the target using the binary range measurements from all nodes.
We first present a framework for target localization with full precision (or infinite-bit) range measurements and follow it with our methods for one-bit data.

\subsection{Localization with Full-Precision Range Estimates}
\label{ssec:Loc_Inf}
Recall the expressions of $d_m$ and $r_m$ in \eqref{eq-2} and \eqref{eq-3}, respectively. Without loss of generality, consider the first ($m=1$) sensor as the reference sensor. The difference between the true range with respect to reference sensor and any of the remaining $m$-th $(m > 1)$ sensor is
\begin{align}
\label{eq-4}
r_m-r_1=d_m-d_1,
\end{align}
Rearranging \eqref{eq-4} as
$
r_m -r_1 + d_1=d_m,
$
and squaring both sides produces\par\noindent\small
\begin{align}
    ((r_m -r_1) + d_1)^2\!=\!d_m^2
    \!=\!(\delta^x_m-\delta^x_1)^2 \!+\! (\delta^y_m-\delta^y_1)^2 \!+\! (\delta^z_m-\delta^z_1)^2,
\end{align}\normalsize
where the last equality follows after substituting $d_m$ from \eqref{eq-2}. Simplifying yields
\par\noindent\small
\begin{align}
\label{eq-6}
&(\delta^x-\delta^x_1) (\delta^x_m-\delta^x_1) + (\delta^y-\delta^y_1) (\delta^y_m-\delta^y_1) + (\delta^z-\delta^z_1) (\delta^z_m-\delta^z_1) \nonumber\\
& + ( r_m -r_1) d_1 =\\ & \dfrac{1}{2}\left[(\delta^x_m-\delta^x_1)^2 + (\delta^y_m-\delta^y_1)^2 + (\delta^z_m-\delta^z_1)^2 - (r_m - r_1)^2\right], \nonumber   
\end{align}\normalsize
which are linear in the target coordinates $\begin{bmatrix} \delta^x \!&\! \delta^y \!&\! \delta^z\end{bmatrix}^T$. Denote the unknown parameter vector
\begin{align}
\THETA = \begin{bmatrix} \delta^x-\delta^x_1 \!&\! \delta^y-\delta^y_1 \!&\! \delta^z-\delta^z_1 \!&\! d_1 \end{bmatrix}^T \in \mathds{R}^{4 \times 1}.
\end{align}
Then, collecting all linear equations specified by \eqref{eq-6} for $m = {2, \cdots, M}$, we obtain the following compact matrix form
\begin{align}
\label{eq-7}
 \GM \THETA  = \HV,
\end{align}
where\par\noindent\footnotesize
\begin{align}
\label{eq-8}
\GM\! = \! \begin{bmatrix}
(\delta^x_2-\delta^x_1) & (\delta^y_2-\delta^y_1) & (\delta^z_2-\delta^z_1) & r_2 -r_1 \\
\vdots & \vdots & \vdots & \vdots \\
(\delta^x_M-\delta^x_1) & (\delta^y_M-\delta^y_1) & (\delta^z_M-\delta^z_1) & r_M - r_1
\end{bmatrix} 
\in \mathbb{R}^{(M-1)\times 4},
\end{align}\normalsize
and
\par\noindent\small
\begin{align}
\label{eq-9}
&\hspace{-4mm}\HV=\frac{1}{2}
\begin{bmatrix}
(\delta^x_2-\delta^x_1)^2 + (\delta^y_2-\delta^y_1)^2 + (\delta^z_2-\delta^z_1)^2 - (r_2-r_1)^2 \\
\vdots \\
(\delta^x_m-\delta^x_1)^2 + (\delta^y_M-\delta^y_1)^2 + (\delta^z_M-\delta^z_1)^2 - (r_M-r_1)^2
\end{bmatrix} \nonumber \\
&~~~~~~~~~~~~~~~~~~~~~~~~~~~~~~~~~~~~~~~~~~~~~~~
\in \mathbb{R}^{(M-1) \times 1}.
\end{align}\normalsize

In practice, every true $m$-th sensor range $r_m$ is unknown. As explained in the previous section, we employ constrained WLS to obtain the estimate $\widehat{r}_m$. Assume $\widehat{r}_m= r_m + e_m$, where $e_m$ is the estimation error due to the receiver noise. Then, the equality in \eqref{eq-7} does not hold and the resulting perturbed system of equations takes the form
\begin{align}
    \label{eq-10}
    \EPSILON = \GM \THETA - \HV,
\end{align}
where $\EPSILON$ denotes the perturbation term.
Assuming $\GM$ is full column rank, the least squares (LS) solution of the system of linear equations in \eqref{eq-10} yields
\begin{align}
\label{eq-12}
\widehat{\THETA}=\GM^{\dagger} \HV.   
\end{align}
Then, the target location is obtained as \par\noindent\small
\begin{align}
\label{full-ran-est}
\begin{bmatrix} \delta^x \!&\! \delta^y \!&\! \delta^z \end{bmatrix}^T = \begin{bmatrix} [\widehat{\THETA}]_1 + \delta^x_1 & [\widehat{\THETA}]_2 + \delta^y_1 & [\widehat{\THETA}]_3 + \delta^z_1 \end{bmatrix}^T.
\end{align}\normalsize
\begin{rmk}
\label{rmk-rev}
Contrary to range estimation, WLS is not applicable for estimating $\THETA$ in \eqref{eq-10} because the covariance matrix of perturbation $\EPSILON$ is unknown. This is apparent from the fact that the covariance matrix of the perturbation term is a function of the variances of the range estimation errors, i.e., $e_1, e_2, \cdots, e_M$, as well as the unknown target location. Under such circumstances, the best choice for the weighting matrix is the identity matrix, which reduces WLS to LS.
\end{rmk}

When the FC receives the full-precision nodal range estimates, i.e., $\widehat{r}_m$ for $1 \leq  m \leq M$, the aforementioned LS solution in \eqref{full-ran-est} is quite effective. However, when the nodal range estimates are quantized to one-bit as in \eqref{Eq-3}, the LS approach is no longer applicable at the FC.
\subsection{Optimal Localization with One-Bit Nodal Range Estimates}
\label{ssec:Loc_one:opt}
We first develop an optimal approach for localization with one-bit quantized range measurements from the $M$ nodes denoted by $\WV= \begin{bmatrix} w_1, w_2, \cdots, w_M \end{bmatrix}^T$. We show that this optimal approach achieves the global minimum.

Consider $\overline{\RV}=\begin{bmatrix} r_2 & r_3 & \cdots & r_M\end{bmatrix}^T \in \mathds{R}^{(M-1) \times 1}$ and denote $\ONEVV$ as a $(M-1) \times 1$ vector with all ones as its elements. Define\par\noindent\small
\begin{align}
\VM=\begin{bmatrix}
(\delta^x_2-\delta^x_1) & (\delta^y_2-\delta^y_1) & (\delta^z_2-\delta^z_1)  \\
\vdots & \vdots & \vdots \\
(\delta^x_M-\delta^x_1) & (\delta^y_M-\delta^y_1) & (\delta^z_M-\delta^z_1) 
\end{bmatrix} \in \mathbb{R}^{(M-1)\times 3},
\end{align}\normalsize
and
\par\noindent\small
\begin{align}
\label{eq-15}
\BV=\dfrac{1}{2}\begin{bmatrix}
(\delta^x_2-\delta^x_1)^2 + (\delta^y_2-\delta^y_1)^2 + (\delta^z_2-\delta^z_1)^2\\
\vdots \\
(\delta^x_M-\delta^x_1)^2 + (\delta^y_M-\delta^y_1)^2 + (\delta^z_M-\delta^z_1)^2
\end{bmatrix} \in \mathbb{R}^{(M-1) \times 1}.
\end{align}\normalsize
Both $\VM$ and $\BV$ are known \textit{a priori}. Then,
\par\noindent\small
\begin{align}
\label{eq-13}
\GM &= \begin{bmatrix} \VM & \overline{\RV}-r_1 \ONEVV \end{bmatrix}, \\
\label{eq-14}
\HV &=  \BV - \dfrac{1}{2} (\overline{\RV}-r_1 \ONEVV) \odot (\overline{\RV}-r_1 \ONEVV).
\end{align}\normalsize
We jointly estimate the unknown $\THETA$ and $\RV$ by solving the optimization
\begin{align}
    \label{eq-joint-range-location}
\begin{array}{ll}
\underset{\RV,\THETA}{\MIN} & \|\GM \THETA - \HV \|^2_2 \\
\SUBJTO & \WV \odot (\RV - \LAMBDA) \succeq \ZEROVV, \\
~ & \RV \succeq \ZEROVV,
\end{array}
\end{align}
where $\LAMBDA=[\lambda_1, \lambda_2, \cdots, \lambda_M]^T$. The first linear constraint in \eqref{eq-joint-range-location}, similar to the formulation in Section \ref{sec:time-delay estimation}, arises because the one-bit quantized data and the elements of $\RV - \LAMBDA$ must share the same sign; and the second constraint indicates that range values are non-negative. 
Reformulate the objective function ${\cal L}(\RV , \THETA) \triangleq \|\GM \THETA - \HV \|^2_2$ as
\par\noindent\footnotesize
\begin{align}
\label{eq-15-2}
{\cal L}(\RV , \THETA) &\triangleq
\left\|\begin{bmatrix} \VM & \overline{\RV}-r_1 \ONEVV\end{bmatrix} \THETA -  \BV + \dfrac{1}{2} (\overline{\RV}-r_1 \ONEVV) \odot (\overline{\RV}-r_1 \ONEVV) \right\|^2_2.
\end{align}\normalsize

When $\RV$ is fixed, the LS solution for $\THETA$ is given by \eqref{eq-12}. Substituting \eqref{eq-12} into \eqref{eq-15-2} yields
\par\noindent\small
\begin{align}
    \label{eq-16}
   {\cal L}(\RV) &= {\cal L}(\RV , \widehat{\THETA}) \triangleq \|\GM\GM^{\dagger} \HV - \HV \|^2_2=  \| \Pi^{\perp}_{\GM} \HV  \|^2_2 \\
   &=\left\|\big[\Pi^{\perp}_{\VM}-\Pi_{\Pi^{\perp}_{\VM} (\overline{\RV}-r_1 \ONEVV)}\big] \big[\BV-\dfrac{1}{2} (\overline{\RV}-r_1 \ONEVV) \odot (\overline{\RV}-r_1 \ONEVV) \big] \right\|^2_2, \nonumber
\end{align}\normalsize
where the last equality 
is obtained by substituting \eqref{eq-13}-\eqref{eq-14} and using $\Pi^{\perp}_{\GM}=\Pi^{\perp}_{\VM}-\Pi_{\Pi^{\perp}_{\VM} (\overline{\RV}-r_1 \ONEVV)}$ following the projection decomposition theorem \cite{yanai2011projection}.
Since $\Pi^{\perp}_{\VM}(\overline{\RV}-r_1 \ONEVV) \in {\cal N}(\VM^H)$, it is easily confirmed that $\Pi^{\perp}_{\VM} \Pi_{\Pi^{\perp}_{\VM} (\overline{\RV}-r_1 \ONEVV)} = \Pi_{\Pi^{\perp}_{\VM} (\overline{\RV}-r_1 \ONEVV)}$ simplifying \eqref{eq-16} to 
\par\noindent\small
\begin{align}
    \label{eq-17}
{\cal L}(\RV) = & \left[\BV-\dfrac{1}{2} (\overline{\RV}-r_1 \ONEVV) \odot (\overline{\RV}-r_1 \ONEVV) \right]^T \left[\Pi^{\perp}_{\VM}-\Pi_{\Pi^{\perp}_{\VM} (\overline{\RV}-r_1 \ONEVV)}\right] \nonumber\\
& \times \left[\BV-\dfrac{1}{2} (\overline{\RV}-r_1 \ONEVV) \odot (\overline{\RV}-r_1 \ONEVV) \right].
\end{align}\normalsize
Expanding $\Pi_{\Pi^{\perp}_{\VM} (\overline{\RV}-r_1 \ONEVV)}$ 
yields
\par\noindent\small
\begin{align}
    \label{eq-18}
\Pi_{\Pi^{\perp}_{\VM} (\overline{\RV}-r_1 \ONEVV)} &= \Pi^{\perp}_{\VM} (\overline{\RV}-r_1 \ONEVV){\Pi^{\perp}_{\VM} (\overline{\RV}-r_1 \ONEVV)}^\dagger\nonumber\\
&=\frac{\Pi^{\perp}_{\VM} (\overline{\RV}-r_1 \ONEVV) (\overline{\RV}-r_1 \ONEVV)^T \Pi^{\perp}_{\VM}} {\|\Pi^{\perp}_{\VM} (\overline{\RV}-r_1 \ONEVV)\|^2_2}.
\end{align}\normalsize
Note that the fact that $\GM$ is full column rank guarantees $\|\Pi^{\perp}_{\VM} (\overline{\RV}-r_1 \ONEVV)\|^2_2 \neq 0$. Substituting \eqref{eq-18} in \eqref{eq-17}, the ${\cal L}(\RV)$ takes the rational form  
$\frac{{\cal F}(\RV)}{{\cal J}(\RV)}$
where ${\cal F}(\RV)$, given in \eqref{eq-21} at the top of the next page,
%
\begin{figure*}[!t]
\par\noindent\small
\begin{align}
\label{eq-21}
{\cal F}(\RV) =
&
\|\Pi^{\perp}_{\VM} (\overline{\RV}-r_1 \ONEVV)\|^2_2 \bigg(\|\Pi^{\perp}_{\VM} \BV\|^2_2 + \dfrac{1}{4} \| \Pi^{\perp}_{\VM} \big[ (\overline{\RV}-r_1 \ONEVV) \odot (\overline{\RV}-r_1 \ONEVV) \big]\|^2_2 
- \BV^T  \Pi^{\perp}_{\VM} \big[ (\overline{\RV}-r_1 \ONEVV) \odot (\overline{\RV}-r_1 \ONEVV) \big]\bigg) - \left(\BV^T  \Pi^{\perp}_{\VM} (\overline{\RV}-r_1 \ONEVV)\right)^2 \nonumber\\
& 
- \dfrac{1}{4} \left( \big[ (\overline{\RV}-r_1 \ONEVV) \odot (\overline{\RV}-r_1 \ONEVV) \big]^T \Pi^{\perp}_{\VM} (\overline{\RV}-r_1 \ONEVV) \right)^2
+ \BV^T \Pi^{\perp}_{\VM} (\overline{\RV}-r_1 \ONEVV) (\overline{\RV}-r_1 \ONEVV)^T \Pi^{\perp}_{\VM} \big[ (\overline{\RV}-r_1 \ONEVV) \odot (\overline{\RV}-r_1 \ONEVV) \big],
\end{align}\normalsize
\hrulefill
\end{figure*}
is a polynomial of degree $6$ and
\begin{align}
\label{eq-21-1}
{\cal J}(\RV) =& \|\Pi^{\perp}_{\VM} (\overline{\RV}-r_1 \ONEVV)\|_2^2,
\end{align}
is a polynomial of degree $2$.
Hence, \eqref{eq-joint-range-location} becomes
\begin{align}
    \label{eq:fracopt}
\begin{array}{ll}
\underset{\RV}{\MIN} & \dfrac{{\cal F}(\RV)}{{\cal J}(\RV)}\\
\SUBJTO & \WV \odot (\RV - \LAMBDA) \succeq \ZEROVV, \\
~ & \RV \succeq \ZEROVV.
\end{array}
\end{align}
The optimization problem in \eqref{eq:fracopt} is non-convex. In order to relax this fractional structure, we decouple the numerator and the denominator as stated in the following theorem.
\begin{theo}\label{theorem-1}
The optimization problem in \eqref{eq:fracopt} is equivalent to
\begin{align}
\label{eq:jsdp}
\begin{array}{ll}
\underset{v, \RV}{\textrm{minimize}} & v \\
\SUBJTO &  v {\cal J}(\RV) -{\cal F}(\RV) \geq 0, \\
  & \WV \odot
(\RV - \LAMBDA) \succeq \ZEROVV, \\
  & \RV \succeq \ZEROVV,
\end{array}
\end{align}
where $v$ is a slack variable.
\end{theo}
\begin{IEEEproof}
See Appendix \ref{App-A}.
\end{IEEEproof}
The objective in the optimization problem \eqref{eq:jsdp} is not rational. However, it is still non-convex because of the polynomial constraint $v {\cal J}(\RV) -{\cal F}(\RV) \geq 0$ of degree $6$.  
To reformulate the problem to an equivalent SDP, we employ Lasserre's multivariate polynomial optimization \cite{lasserre2001global}. 
\begin{Def}[Monomial basis of degree $p$]
\label{definition-1}
The vector $\GV_p(\UV)$ is called the monomial basis of degree $p$ if it contains all monomials $u^{\nu_1}_1 u^{\nu_2}_2 \cdots u^{\nu_q}_q$ such that $\sum_{i=1}^q \nu_i \leq p$ with $\nu_i$'s being integers. 
\end{Def}
For example, $\GV_2(u_1, u_2)$ is the monomial basis of degree $2$ if
\begin{align}
 \GV_2([u_1, u_2]^T) = \begin{bmatrix} 1 & u_1 & u_2 & u_1^2 & u_1 u_2 & u_2^2
 \end{bmatrix}^T.
\end{align}
To parametrize the first constraint of \eqref{eq:jsdp}, substituting \eqref{eq-21}-\eqref{eq-21-1} in $v {\cal J}(\RV)-{\cal F}(\RV)$, and expanding the resulting equation, we obtain \eqref{eq-21-exp} given at the top of the next page, 
\begin{figure*}[!t]
\par\noindent\small
\begin{align}
\label{eq-21-exp}
&v{\cal J}(\RV) - {\cal F}(\RV) \!=\! \sum_{m=1}^{M} \psi_{mm} r_m^2 v + ( \kappa_m^2 \!-\! \chi \psi_{mm})r_m^2 + \underset{m \neq n}{\sum_{m=1}^{M}\sum_{n=1}^{M}} \psi_{mn} r_m r_n v + \dfrac{(\psi_{mn}^2 \!-\! \psi_{mm} \psi_{nn})}{4}(r_m^4r_n^2 \!-\! r_m^3r_n^3) \nonumber\\
& + (\kappa_m \psi_{mn} \!-\! \psi_{mm} \psi_{nn} ) (r_m^3 r_n \!-\! r_m^2 r_n^2 ) \!+\! (\kappa_m  \kappa_n \!-\! \chi \psi_{mn} )r_m r_n + \underset{m \neq n \neq k}{\sum_{m=1}^{M}\sum_{n=1}^{M} \sum_{k=1}^{M}} \dfrac{( \psi_{mn}\psi_{mk} - \psi_{mm} \psi_{nk})}{4} (r_m^4r_n r_k-2r_m^3r_n^2r_k + r_m^2r_n^2r_k^2) \nonumber\\ 
&+ (\psi_{mm} \psi_{nk} - \psi_{mn} \kappa_k)r_m^2r_nr_k + \underset{m \neq n \neq k \neq q}{\sum_{m=1}^{M}\sum_{n=1}^{M} \sum_{k=1}^{M} \sum_{q=1}^{M}} \dfrac{( \psi_{mk}\psi_{nq} \!-\! \psi_{mn} \psi_{kq})}{4}r_m^2r_n^2r_kr_q + \underset{m \neq n \neq k}{\sum_{m=2}^{M}\sum_{n=2}^{M} \sum_{k=2}^{M}} (\psi_{1m}\psi_{nk} \!-\! 3\psi_{mn} \psi_{1k}) r_1^3 r_m r_n r_k \nonumber\\
& + (4 \psi_{mn} \psi_{mk} + 3\psi_{mm}\psi_{nk}) r_m^3 r_n r_k r_1 + (\psi_{mn} \kappa_k - 2 \psi_{mm} \psi_{nk})r_m r_n r_k r_1 + 3 \underset{m \neq n \neq k \neq q}{\sum_{m=2}^{M}\sum_{n=2}^{M} \sum_{k=2}^{M} \sum_{q=2}^{M}} \psi_{mn} \psi_{kq} r_m^2 r_n r_k r_q,
\end{align}\normalsize
\hrulefill
\end{figure*}
where
\begin{align}
\psi_{mn} & \!=\! \left\{\begin{array}{ll}
           \sum_{i=1}^{M-1}\sum_{j=1}^{M-1} [\Pi^{\perp}_{\VM}]_{i,j}, & \text{if}~~ m=n=1, \\
            \big[\Pi^{\perp}_{\VM}\big]_{m-1,m-1}, & \text{if}~~ 2 \leq m =n \leq M, \\
            - \sum_{i=1}^{M-1} [\Pi^{\perp}_{\VM}]_{i,m-1}, & \text{if}~~ m = 1, 2 \leq n \leq M \\
            \big[\Pi^{\perp}_{\VM}\big]_{m-1,n-1}, & \text{if}~~ 2 \leq m \neq n \leq M,
          \end{array}
\right.\\
\kappa_{m} & \!=\!\left\{\begin{array}{ll}
           - \sum_{i=1}^{M-1}\sum_{j=1}^{M-1} [\Pi^{\perp}_{\VM}]_{i,j} [\BV]_j, & \text{if}~~ m = 1, \\
            \sum_{i=1}^{M-1} [\Pi^{\perp}_{\VM}]_{i,m-1} [\BV]_j, & \text{if}~~ 2 \leq m \leq M,
          \end{array}
\right.
\end{align}
and $\chi=\|\Pi^{\perp}_{\VM} \BV\|^2_2$. Using Definition~\ref{definition-1}, we parameterize the polynomial in the first constraint of \eqref{eq:jsdp}
as
\begin{align}
v {\cal J}(\RV)-{\cal F}(\RV) = \PHI^T \GV_6([\RV,v]^T)
\end{align}
where $\PHI$ is the vector of the coefficients corresponding to the monomial basis $\GV_6([\RV,v]^T)$, which is readily obtained from \eqref{eq-21-exp}. We state the SDP equivalent of \eqref{eq:jsdp} in the following theorem.
\begin{theo}
\label{theorem-2}
Given the scalars $r_1$, $r_2$, $\cdots$, $r_{\Scale[0.5]{M}}$ and integers $\{\nu_i\}_{i=1}^M$, define ${\cal K}: \mathds{R}^{M+1} \to \mathds{R}$ as ${\cal K}(r_1^{\nu_1} r_2^{\nu_2} \cdots r_{\Scale[0.5]{M}}^{\nu_{\Scale[0.5]{M}}} v^{\nu_{\Scale[0.5]{M\!+\!1}}}) = \mu_{\nu_1 \nu_2 \cdots \nu_{\Scale[0.5]{M\!+\!1}}}$ such that ${\cal K}(1) \!=\! \mu_{0 0 \cdots0} \!=\!  1$. Construct the matrices
\par\noindent\footnotesize
\begin{align}
&\TM_{p}(\MU) \!=\! {\cal K}\big(\GV_p([\RV,v]^T) \GV_p^T([\RV,v]^T)\big), \\
&\TM^m_{p-1}(\MU) \!=\!\\
&\left\{\begin{array}{ll} \hspace{-2mm}
           {\cal K}\big(\GV_{p-1}([\RV,v]^T) \GV_{p-1}^T([\RV,v]^T) w_m(r_m \!-\! \lambda_m) \big), \!&\! \text{if}~ 1 \!\leq\! m \!\leq\! M , \\
            \hspace{-2mm} {\cal K}\big(\GV_{p-1}([\RV,v]^T) \GV_{p-1}^T([\RV,v]^T) r_m\big), \!&\! \text{if}~ M+1 \!\leq\! m \!\leq\! 2M,\\
            \hspace{-2mm} {\cal K}\big(\GV_{p-1}([\RV,v]^T) \GV_{p-1}^T([\RV,v]^T) (v_{\rm max}-v)\big) \!&\! \text{if}~ m \!=\! 2M+1,
          \end{array}
\right.  \nonumber
\end{align}\normalsize
and\par\noindent\small
\begin{align}
&\TM_{p-3}(\MU) \!=\! {\cal K}\big(\GV_{p-3}([\RV,v]^T) \GV_{p-3}^T([\RV,v]^T) \PHI^T \GV_6([\RV,v]^T)\big).
\end{align}\normalsize
Then, there exists an integer $p \geq 3$ for which the optimization problem \eqref{eq:jsdp} is equivalent to 
\begin{align}
 \begin{array}{ll}
 \label{eq-opt-sdp}
\underset{\MU}{\textrm{minimize}} & \MU_{00\cdots01} \\
\SUBJTO &  \TM_{p}(\MU) \succeq \ZEROVV, \\
  & \TM_{p-3}(\MU) \succeq \ZEROVV, \\
  & \TM_{p-1}^m(\MU) \succeq \ZEROVV, ~~ 1 \leq m \leq 2M+1,
\end{array}  
\end{align}
such that the minimizer of \eqref{eq:jsdp} is \par\noindent\small
\begin{align}
[r_1^\star, r_2^\star, \cdots, r_M^\star, v^\star]^T = [\mu_{10\cdots00}^\star, \mu_{01\cdots00}^\star, \cdots, \mu_{00\cdots10}^\star, \mu_{00\cdots01}^\star]^T.
\end{align}
\end{theo}
\begin{IEEEproof}
See Appendix \ref{App-B}.
\end{IEEEproof}
\begin{rmk}
Note that the number of optimization variables in \eqref{eq-opt-sdp} is equal to $\binom{M+2p+1}{2p}$ which could be very large even for moderate values of the number of sensors $M$ and the relaxation order $p$. Therefore, even though this method is able to attain the global minimum, it could become computationally expensive in the practical scenarios.
\end{rmk}
\subsection{Sub-Optimal Localization with One-Bit Nodal Range Estimates}
\label{ssec:Loc_one:sopt}
It is possible to reduce the computational complexity of the Lasserre's SDP method by trading off the optimality.
We now present such a sub-optimal approach by iteratively solving \eqref{eq-joint-range-location} through alternating minimizations over $\THETA$, $r_1$ and $\overline{\RV}$. Although this method, that we call ANTARES standing for iterative joint r\textit{AN}ge-\textit{TAR}get location {\textit ES}timation, achieves only a local minimum, its computationally efficiency is significantly higher than SDP. 

Denote $\THETA^{(k)}$, $r_1^{(k)}$ and $\overline{\RV}^{(k)}$ to be the values of the parameters $\THETA$, $r_1$ and $\overline{\RV}$ at the $k$-th iteration, respectively. Given $\THETA^{(k)}$ and $r_1^{(k)}$, using \eqref{eq-15-2}, the problem in \eqref{eq-joint-range-location} with respect to $\overline{\RV}$ at the $(k+1)$-th iteration becomes
\par\noindent\small
\begin{align}
    \label{eq-sub-optimal}
    \hspace{-2mm}
\begin{array}{ll}
\underset{\overline{\RV}}{\MIN} & \sum_{m=2}^M \left(\dfrac{(r_m-r_1^{(k)})^2}{2} + [\THETA^{(k)}]_4 (r_m -r_1^{(k)}) + \zeta_m^{(k)}\right)^2\\
\SUBJTO & \begin{array}{ll} w_m (r_m - \lambda_m) \geq 0, & 2 \leq m \leq M,\end{array} \\
 & \begin{array}{ll} r_m \geq 0, & \qquad \qquad \quad 2 \leq m \leq M, \end{array}
\end{array} \hspace{-3mm}
\end{align}
where $\zeta_m^{(k)} \!=\! [\VM \overline{\THETA}^{(k)}]_{m-1} - [\BV]_{m-1}$ with $\overline{\THETA}^{(k)} \!=\! \begin{bmatrix} [\THETA^{(k)}]_1 \!&\! [\THETA^{(k)}]_2 \!&\! [\THETA^{(k)}]_3 \end{bmatrix}^T$. The global minimizer of \eqref{eq-sub-optimal} gives the update of $\overline{\RV}^{(k)}$ as $\overline{\RV}^{(k+1)}$ to be used in the next iteration. Observe this optimization problem is separable in $r_2, r_3, \cdots, r_M$. Hence, we convert it into $M-1$ parallel optimization problems, each of which is 
\begin{align}
    \label{eq-sopt-rm}
    \hspace{-2mm}
\begin{array}{ll}
\underset{r_m}{\MIN} & \frac{1}{4} r_m^4 + \beta_m^{(k)} r_m^3 + \varsigma_m^{(k)}  r_m^2 + \omega_m^{(k)} r_m + \eta_m^{(k)} \\
\SUBJTO & w_m (r_m - \lambda_m) \geq 0, \\
~ & r_m \geq 0,
\end{array}\hspace{-2mm}
\end{align}
where \par\noindent\small
\begin{subequations}
\begin{align}
\hspace{-2mm}\beta_m^{(k)} =& [\THETA^{(k)}]_4 - r_1^{(k)}, \\
\hspace{-2mm}\varsigma_m^{(k)}  =& \frac{3(r_1^{(k)})^2}{2} - 3 [\THETA^{(k)}]_4 r_1^{(k)}+ ([\THETA^{(k)}]_4)^2 + \zeta_m^{(k)},\\
\hspace{-2mm}\omega_m^{(k)} =& -(r_1^{(k)})^3 \!+\! 3 [\THETA^{(k)}]_4 (r_1^{(k)})^2 \!-\! 2 \left(([\THETA^{(k)}]_4)^2 \!+\! \zeta_m^{(k)}\right) r_1^{(k)} \nonumber\\
&\!+\! 2 [\THETA^{(k)}]_4 \zeta_m^{(k)},\\
\hspace{-2mm}\eta_m^{(k)}  =& \frac{(r_1^{(k)})^4}{4} \!-\! [\THETA^{(k)}]_4 (r_1^{(k)})^3 \!+\! \left(([\THETA^{(k)}]_4)^2 \!+\! \zeta_m^{(k)}\right) (r_1^{(k)})^2 \nonumber\\
&\!-\! 2 [\THETA^{(k)}]_4 \zeta_m^{(k)} r_1^{(k)} + (\zeta_m^{(k)})^2.
\end{align}
\end{subequations}
\normalsize
Since the objective and constraints in \eqref{eq-sopt-rm} are differentiable, the global minimizer of \eqref{eq-sopt-rm} belongs to a set of points which satisfy the following Karush-Kuhn-Tucker (KKT) conditions \cite{boyd2004convex}:
\par\noindent\small
\begin{subequations}
\begin{align}
\label{eq-kkt-1}
&r_m^3 + 3\beta_m^{(k)} r_m^2 + 2 \varsigma_m^{(k)}  r_m + \omega_m^{(k)} - \varrho_1 w_m - \varrho_2 = 0, \\
\label{eq-kkt-2}
&w_m (r_m - \lambda_m) \geq 0, \\
\label{eq-kkt-3}
&r_m \geq 0, \\
\label{eq-kkt-4}
&\varrho_1 w_m (r_m - \lambda_m) = 0,\\
\label{eq-kkt-5}
&\varrho_2 r_m = 0,\\
\label{eq-kkt-6}
& \varrho_1 \geq 0, \\
\label{eq-kkt-7}
& \varrho_2 \geq 0.
\end{align}
\end{subequations}
\normalsize
where $\varrho_1$ and $\varrho_2$ are the KKT multipliers.
From \eqref{eq-kkt-2}-\eqref{eq-kkt-7}, there are three possibilities:
\begin{enumerate}[label=(\roman*)]
\item $\varrho_1 > 0$ and $\varrho_2=0$: From \eqref{eq-kkt-5}, under this condition, $r_m$ must be equal to $\lambda_m$. Considering $r_m = \lambda_m$ and $\varrho_2=0$, it follows from \eqref{eq-kkt-1} that
\begin{align}
    \label{eq-pos1-kkt}
    \varrho_1 = w_m (\lambda_m^3 + 3\beta_m^{(k)} \lambda^2_m + 2 \varsigma_m^{(k)}  \lambda_m + \omega_m^{(k)}).
\end{align}
Further, from $\varrho_1>0$, the point $r_m = \lambda_m$ satisfies the KKT conditions if 
\begin{align}
\label{eq-pos1-min-con}
w_m (\lambda_m^3 + 3\beta_m^{(k)} \lambda^2 + 2 \varsigma_m^{(k)}  \lambda_m + \omega_m^{(k)}) > 0.
\end{align}
\item $\varrho_1 = 0$ and $\varrho_2 > 0$: From \eqref{eq-kkt-6}, $r_m$ must be zero under this scenario. Considering $r_m = 0$ and $\varrho_1=0$, it follows from \eqref{eq-kkt-1} and \eqref{eq-kkt-2} that $\varrho_2 = \omega_m^{(k)}$ and $w_m \leq 0$. Hence, when $\varrho_2 >0$, the point $r_m = 0$ satisfies the KKT conditions if 
\begin{align}
\label{eq-pos2-min-con}
\left\{\begin{array}{l}
\omega_m^{(k)} > 0,\\
w_m \leq 0.
\end{array}
\right.
\end{align}
\item $\varrho_1 = 0$ and $\varrho_2 = 0$: Under this scenario, the KKT conditions imply that $r_m$ must be equal to the non-negative real roots of the following cubic equation
\begin{align}
\label{eq-cubic}
   r_m^3 + 3\beta_m^{(k)} r_m^2 + 2 \varsigma_m^{(k)}  r_m + \omega_m^{(k)} = 0,
\end{align}
which satisfy \eqref{eq-kkt-2}.
The roots of \eqref{eq-cubic} are given by
\begin{align}
\label{eq-cubic-root}
\digamma_q = - \frac{1}{3}\left(3\beta_m^{(k)}+\xi^q \Delta_2+\frac{\Delta_0}{\xi^q \Delta_2}\right), \quad q \in \{0,1,2\},
\end{align}
where $\xi = \frac{-1+\J\sqrt{3}}{2}$, $\Delta_2 = \sqrt[3]{\frac{\Delta_1 \pm \sqrt{\Delta_1^2 - 4 \Delta_0^3}}2}$, $\Delta_0 = 9(\beta_m^{(k)})^2 - 6 \varsigma_m^{(k)} $ and $\Delta_1 = 54 (\beta_m^{(k)})^3 - 54\beta_m^{(k)}\varsigma_m^{(k)}  + 27 \omega_m^{(k)}$. Further, it is well-known that amongst the KKT-compatible non-negative real roots of \eqref{eq-cubic}, only those which also satisfy the following second-order sufficient condition 
\begin{align}
\label{eq-suf-con}
3\digamma_q^2 + 6 \beta_m^{(k)} \digamma_q + 2 \varsigma_m^{(k)}  \geq 0,
\end{align}
act as the minimizers of \eqref{eq-sopt-rm} \cite{boyd2004convex}.
As a result, we only consider the non-negative real root of \eqref{eq-cubic} for which
\eqref{eq-kkt-2} and \eqref{eq-suf-con} hold true.
\end{enumerate}
Accordingly, the set of points which are the minimizers of \eqref{eq-sopt-rm} is derived by following (i) to (iii) above. Then, the global minimizer of \eqref{eq-sopt-rm} is the point in this set at which the value of the objective in \eqref{eq-sopt-rm} is the smallest.

Once $\overline{\RV}^{(k+1)}$ is found, the problem \eqref{eq-joint-range-location} with respect to $r_1$ at the $(k+1)$-th iteration is cast as
\begin{align}
    \label{eq-sopt-r1}
\begin{array}{ll}
\underset{r_1}{\MIN} & \frac{1}{4} r_1^4 + \beta_1^{(k)} r_1^3 + \varsigma_1^{(k)} r_1^2 + \omega_1^{(k)} r_1 + \eta_1^{(k)} \\
\SUBJTO & w_1 (r_1 - \lambda_1) \geq 0, \\
~ & r_1 \geq 0,
\end{array}
\end{align}
where
\par\noindent\footnotesize
\begin{subequations}
\begin{align}
\beta_1^{(k)} = & \frac{-1}{M-1}\sum_{m=2}^{M} r_m^{(k+1)} - [\THETA^{(k)}]_4,\\
\varsigma_1^{(k+1)} = & \frac{1}{M-1}\sum_{m=2}^M \frac{3}{2} (r_m^{(k+1)})^2 + 3 [\THETA^{(k)}]_4 r_m^{(k+1)} + \zeta_m^{(k)} +  ([\THETA^{(k)}]_4)^2,\\
\omega_1^{(k)} = & \frac{-1}{M-1} \sum_{m=2}^M \!(r_m^{(k+1)})^3 \!+\! 3 [\THETA^{(k)}]_4 (r_m^{(k+1)})^2 \nonumber\\
& \!+\! 2 \left(([\THETA^{(k)}]_4)^2 \!+\! \zeta_m^{(k)}\right) r_m^{(k+1)} \!+\! 2 [\THETA^{(k)}]_4 \zeta_m^{(k)},\\
\eta_1^{(k)} = & \frac{1}{M-1} \sum_{m=2}^M \frac{(r_m^{(k+1)})^4}{4} \!+\! [\THETA^{(k)}]_4 (r_m^{(k+1)})^3 \nonumber\\
& \!+\! \left(([\THETA^{(k)}]_4)^2  \!+\! \zeta_m^{(k)}\right) (r_m^{(k+1)})^2 \!+\! 2 [\THETA^{(k)}]_4 \zeta_m^{(k)} r_m^{(k+1)} \!+\! (\zeta_m^{(k)})^2.
\end{align}
\end{subequations}
\normalsize
The global minimizer of \eqref{eq-sopt-r1} is attained by following a procedure similar to that of \eqref{eq-sopt-rm}. From $\overline{\RV}^{(k+1)}$ and $r_1^{(k+1)}$, the update of $\THETA^{(k)}$ at $(k+1)$-th iteration is
\begin{align}
\label{eq-theta-update}
\THETA^{(k+1)}=\GM^{{\dagger}^{(k+1)}} \HV^{(k+1)},
\end{align}
where $\GM^{{\dagger}^{(k+1)}}$ and $\HV^{(k+1)}$ are computed by substituting $\overline{\RV}^{(k+1)}$ and $r_1^{(k+1)}$ for $\overline{\RV}$ and $r_1$ in \eqref{eq-8} and \eqref{eq-9}, respectively.

Algorithm \ref{alg-2} summarizes the steps of aforementioned ANTARES for joint estimation of $\THETA$ and $\RV$. Note that each iteration of ANTARES requires solving one-dimensional optimizations, each of which has a closed-form solution. Further, the optimizations with respect to $r_2, r_3, \cdots, r_m$ are solved in parallel at each iteration. Hence, ANTARES is computationally highly efficient compared to \eqref{eq-opt-sdp}.
%
\begin{algorithm}
\caption{Iterative joint r\textit{an}ge-\textit{tar}get location {\textit es}timation (ANTARES)}
\begin{algorithmic}[1]
\qinput 
one-bit samples $\WV$, threshold vector $\LAMBDA$, optimality tolerance parameters $\varepsilon_1$ and $\varepsilon_2$. 
\qoutput Target location estimate $\widehat{\THETA}$, range estimate $\widehat{\RV}$. 
    \State {\bf Initialization:} Set $k=0$, $\THETA^{(0)} \in \mathds{R}^{4 \times 1}$ arbitrarily and $r_1^{(0)} \geq 0$ such that $w_1 (r_1^{(0)} -\lambda_1) > 0$.
    \While{ $\|\THETA^{(k+1)} - \THETA^{(k)}\|_2^2 \geq \varepsilon_1$ and $\|\RV^{(k+1)} - \RV^{(k)}\|_2^2 \geq \varepsilon_2$}
    \If{$2 \leq m \leq M$}
    \State $\mathds{S} \gets \{\varnothing\}$.
    \If{\eqref{eq-pos1-min-con} is fulfilled}
    \State $\mathds{S} \gets \{\lambda_m\} \cup \mathds{S}$.
    \Else
    \State $\mathds{S} \gets \mathds{S}$.
    \EndIf
    \If{\eqref{eq-pos2-min-con} is fulfilled}
    \State $\mathds{S} \gets \{0\} \cup \mathds{S}$.
    \Else
    \State $\mathds{S} \gets \mathds{S}$.
    \EndIf
    \For{$q \gets 0$ to $2$}
    \State $\mathds{D} \gets \{\varnothing\}$.
    \State Find $\digamma_q$ from \eqref{eq-cubic-root}.
    \If{\small $w_m (\digamma_q - \lambda_m) \geq 0$, $\digamma_q \geq 0$, $\IM\{\digamma_q\}=0$ and $3\digamma_q^2 + 6 \beta_m^{(k)} \digamma_q + 2 \varsigma_m^{(k)}  \geq 0$}\normalsize
    \State $\mathds{D} \gets \mathds{D} \cup \digamma_q$.
    \EndIf
    \EndFor
     \State $\mathds{S} \gets \mathds{D} \cup \mathds{S}$.
      \State Find $r_{\rm opt} \in \mathds{S}$ at which the objective of \eqref{eq-sub-optimal} is minimized.
      \State $r_m^{(k+1)} \gets r_{\rm opt}$.
        \EndIf
        \State Follow steps $4$-$17$ to solve \eqref{eq-sopt-r1} for $r_1^{(k+1)}$. 
        \State $\THETA^{(k+1)} \gets \GM^{{\dagger}^{(k+1)}} \HV^{(k+1)}$.
    \EndWhile
    \State $\widehat{\THETA} = \THETA^{(k+1)}$ and $\widehat{\RV} = \RV^{(k+1)}$.
\end{algorithmic}
\label{alg-2}
\end{algorithm}
%
\subsection{CRB for Localization with One-Bit Nodal Range Estimates}
\label{sec:crb}
We employ the CRB as a benchmark for assessing the estimation performance of the proposed optimal and sub-optimal algorithms. This is also useful for demonstrating the performance loss of one-bit quantization over the unquantized processing.

Assume that the estimation error term in $\widehat{r}_m = r_m + e_m$, i.e., $e_m$, follows a zero-mean Gaussian distribution with variance $\upsilon_m^2$, $1 \leq m \leq M$. Then, $\widehat{r}_m$ is distributed as a Gaussian random variable with mean $r_m$ and variance $\upsilon_m^2$, $1 \leq m \leq M$. The $\widehat{r}_1, \widehat{r}_2, \cdots, \widehat{r}_M$ are statistically independent. Hence, the conditional probability density function of $\WV$ given $\QV = [\delta^x, \delta^y, \delta^z, d_0, \upsilon_1, \upsilon_2, \cdots, \upsilon_M]^T \in \mathds{R}^{(M+4) \times 1}$ is
\begin{align}
\label{eq-CRB-1}
f(\WV \mid \QV) = \prod_{m=1}^M \Phi(\frac{w_m(r_m-\lambda_m)}{\upsilon_m}),
\end{align}
where $\Phi(x) = \frac{1}{\sqrt{2}} \int_{\infty}^x e^{-u^2/2} {\rm d}u$. The CRB is the inverse of the Fisher Information Matrix (FIM) $\mathbf{I}(\QV)$, whose $(i,j)$-th element is \cite{kay1993fundamentals}
\begin{align}
\label{eq-CRB-2}
[\mathbf{I}(\QV)]_{i,j} = \EX\left\{ \frac{\partial \log f(\WV \mid \QV)}{\partial [\QV]_i} \frac{\partial \log f(\WV \mid \QV)}{\partial [\QV]_j}\right\}.
\end{align}
From \eqref{eq-CRB-1}, \eqref{eq-2} and \eqref{eq-3}, the partial derivatives of the log-likelihood $\log f(\WV \mid \QV)$ are
\begin{align}
\label{eq-CRB-3}
\frac{\partial \log f(\WV \mid \QV)}{\partial \delta^x} &= \frac{1}{\sqrt{2\pi}}\sum_{m=1}^M \frac{w_m(\delta^x - \delta^x_m) e^{-\frac{(r_m - \lambda_m)^2}{2 \upsilon^2_m}}}{\upsilon_m d_m \Phi(\frac{w_m(r_m-\lambda_m)}{\upsilon_m})},\\
\label{eq-CRB-4}
\frac{\partial \log f(\WV \mid \QV)}{\partial \delta^y} &= \frac{1}{\sqrt{2\pi}}\sum_{m=1}^M \frac{w_m(\delta^y - \delta^y_m) e^{-\frac{(r_m - \lambda_m)^2}{2 \upsilon^2_m}}}{\upsilon_m d_m \Phi(\frac{w_m(r_m-\lambda_m)}{\upsilon_m})},\\
\label{eq-CRB-5}
\frac{\partial \log f(\WV \mid \QV)}{\partial \delta^z} &= \frac{1}{\sqrt{2\pi}}\sum_{m=1}^M \frac{w_m(\delta^z - \delta^z_m) e^{-\frac{(r_m - \lambda_m)^2}{2 \upsilon^2_m}}}{\upsilon_m d_m \Phi(\frac{w_m(r_m-\lambda_m)}{\upsilon_m})},\\
\label{eq-CRB-6}
\frac{\partial \log f(\WV \mid \QV)}{\partial d_0} &= \frac{1}{\sqrt{2\pi}}\sum_{m=1}^M \frac{w_m e^{-\frac{(r_m - \lambda_m)^2}{2 \upsilon^2_m}}}{\upsilon_m d_m \Phi(\frac{w_m(r_m-\lambda_m)}{\upsilon_m})},\\
\label{eq-CRB-7}
\frac{\partial \log f(\WV \mid \QV)}{\partial \upsilon_m} &= - \frac{w_m (r_m - \lambda_m) e^{-\frac{(r_m - \lambda_m)^@}{2 \upsilon^2_m}}}{\upsilon_m^2 \Phi(\frac{w_m(r_m-\lambda_m)}{\upsilon_m})},\;1 \!\leq\! m \!\leq\! M.
\end{align}
Inserting \eqref{eq-CRB-3} to \eqref{eq-CRB-7} into \eqref{eq-CRB-3} and exploiting the statistical independence of $w_1, w_2, \cdots, w_M$, the elements of the FIM are
\begin{align}
\label{eq-CRB-8}
&[\mathbf{I}(\QV)]_{1,1}\!=\! \sum_{m=1}^M \frac{(\delta^x_m- \delta^x)^2}{2\pi \upsilon_m^2 d_m^2}\bigg[\frac{e^{-\frac{(r_m - \lambda_m)^2}{\upsilon^2_m}}}{\Phi(\frac{r_m-\lambda_m}{\upsilon_m})}\!+\!\frac{e^{-\frac{(r_m - \lambda_m)^2}{\upsilon^2_m}}}{\Phi(\frac{-r_m+\lambda_m}{\upsilon_m})}\bigg],\\
\label{eq-CRB-9}
&[\mathbf{I}(\QV)]_{2,2} \!=\! \sum_{m=1}^M \frac{(\delta^y_m- \delta^y)^2}{2\pi \upsilon_m^2 d_m^2}\bigg[\frac{e^{-\frac{(r_m - \lambda_m)^2}{\upsilon^2_m}}}{\Phi(\frac{r_m-\lambda_m}{\upsilon_m})}\!+\!\frac{e^{-\frac{(r_m - \lambda_m)^2}{\upsilon^2_m}}}{\Phi(\frac{-r_m+\lambda_m}{\upsilon_m})}\bigg],\\
\label{eq-CRB-10}
&[\mathbf{I}(\QV)]_{3,3} \!=\! \sum_{m=1}^M \frac{(\delta^z_m- \delta^z)^2}{2\pi \upsilon_m^2 d_m^2}\bigg[\frac{e^{-\frac{(r_m - \lambda_m)^2}{\upsilon^2_m}}}{\Phi(\frac{r_m-\lambda_m}{\upsilon_m})}\!+\!\frac{e^{-\frac{(r_m - \lambda_m)^2}{\upsilon^2_m}}}{\Phi(\frac{-r_m+\lambda_m}{\upsilon_m})}\bigg],\\
\label{eq-CRB-11}
&[\mathbf{I}(\QV)]_{1,2} \!=\! \sum_{m=1}^M \frac{(\delta^x_m- \delta^x)(\delta^y_m- \delta^y)e^{-\frac{(r_m - \lambda_m)^2}{\upsilon^2_m}}}{2\pi \upsilon_m^2 d_m^2}\nonumber\\  &\;\;\;\;\;\;\;\;\;\;\;\;\times\bigg[\frac{1}{\Phi(\frac{r_m-\lambda_m}{\upsilon_m})}\!+\!\frac{1}{\Phi(\frac{-r_m+\lambda_m}{\upsilon_m})}\bigg],\\
\label{eq-CRB-12}
&[\mathbf{I}(\QV)]_{1,3} \!=\! \sum_{m=1}^M \frac{(\delta^x_m- \delta^x)(\delta^z_m- \delta^z)e^{-\frac{(r_m - \lambda_m)^2}{\upsilon^2_m}}}{2\pi \upsilon_m^2 d_m^2}\nonumber\\  &\;\;\;\;\;\;\;\;\;\;\;\;\times\bigg[\frac{1}{\Phi(\frac{r_m-\lambda_m}{\upsilon_m})}\!+\!\frac{1}{\Phi(\frac{-r_m+\lambda_m}{\upsilon_m})}\bigg], \\
&[\mathbf{I}(\QV)]_{2,3} \!=\! \sum_{m=1}^M \frac{(\delta^y_m- \delta^y)(\delta^z_m- \delta^z)e^{-\frac{(r_m - \lambda_m)^2}{\upsilon^2_m}}}{2\pi \upsilon_m^2 d_m^2}\nonumber\\  &\;\;\;\;\;\;\;\;\;\;\;\;\times\bigg[\frac{1}{\Phi(\frac{r_m-\lambda_m}{\upsilon_m})}\!+\!\frac{1}{\Phi(\frac{-r_m+\lambda_m}{\upsilon_m})}\bigg],\\
&[\mathbf{I}(\QV)]_{4,4} \!=\! \sum_{m=1}^M \frac{e^{-\frac{(r_m - \lambda_m)^2}{\upsilon^2_m}}}{2\pi \upsilon_m^2 d_m^2}\bigg[\frac{1}{\Phi(\frac{r_m-\lambda_m}{\upsilon_m})}\!+\!\frac{1}{\Phi(\frac{-r_m+\lambda_m}{\upsilon_m})}\bigg],
\end{align}
\begin{align}
&[\mathbf{I}(\QV)]_{1,4} \!=\! \sum_{m=1}^M \frac{(\delta^x- \delta^x_m)}{2\pi \upsilon_m^2 d_m^2}\bigg[\frac{e^{-\frac{(r_m - \lambda_m)^2}{\upsilon^2_m}}}{\Phi(\frac{r_m-\lambda_m}{\upsilon_m})}\!+\!\frac{e^{-\frac{(r_m - \lambda_m)^2}{\upsilon^2_m}}}{\Phi(\frac{-r_m+\lambda_m}{\upsilon_m})}\bigg],\\
&[\mathbf{I}(\QV)]_{2,4} \!=\! \sum_{m=1}^M \frac{(\delta^y- \delta^y_m)}{2\pi \upsilon_m^2 d_m^2}\bigg[\frac{e^{-\frac{(r_m - \lambda_m)^2}{\upsilon^2_m}}}{\Phi(\frac{r_m-\lambda_m}{\upsilon_m})}\!+\!\frac{e^{-\frac{(r_m - \lambda_m)^2}{\upsilon^2_m}}}{\Phi(\frac{-r_m+\lambda_m}{\upsilon_m})}\bigg],\\
&[\mathbf{I}(\QV)]_{3,4} \!=\! \sum_{m=1}^M \frac{(\delta^z- \delta^z_m)}{2\pi \upsilon_m^2 d_m^2}\bigg[\frac{e^{-\frac{(r_m - \lambda_m)^2}{\upsilon^2_m}}}{\Phi(\frac{r_m-\lambda_m}{\upsilon_m})}\!+\!\frac{e^{-\frac{(r_m - \lambda_m)^2}{\upsilon^2_m}}}{\Phi(\frac{-r_m+\lambda_m}{\upsilon_m})}\bigg], \\
&[\mathbf{I}(\QV)]_{m+4,m+4} \!=\! \frac{(r_m- \lambda_m)^2}{2\pi \upsilon_m^4}\bigg[\frac{e^{-\frac{(r_m - \lambda_m)^2}{\upsilon^2_m}}}{\Phi(\frac{r_m-\lambda_m}{\upsilon_m})}\nonumber\\
&\;\;\;\;\;\;\;\;\;\;\;\;\;\;\;\;\;\;\;\;\;\;\;\; \!+\!\frac{e^{-\frac{(r_m - \lambda_m)^2}{\upsilon^2_m}}}{\Phi(\frac{-r_m+\lambda_m}{\upsilon_m})}\bigg], \;1 \!\leq\! m \!\leq\! M,\\
&[\mathbf{I}(\QV)]_{m+4,m'+4} = 0, \;1 \!\leq\! m \neq m' \!\leq\! M,\\\
&[\mathbf{I}(\QV)]_{1,m+4} \!=\! \frac{(\delta^x_m-\delta^x)(r_m- \lambda_m)}{2\pi \upsilon_m^4}\bigg[\frac{e^{-\frac{(r_m - \lambda_m)^2}{\upsilon^2_m}}}{\Phi(\frac{r_m-\lambda_m}{\upsilon_m})}\nonumber\\
&\;\;\;\;\;\;\;\;\;\;\;\;\;\;\;\;\;\;\;\; \!+\!\frac{e^{-\frac{(r_m - \lambda_m)^2}{\upsilon^2_m}}}{\Phi(\frac{-r_m+\lambda_m}{\upsilon_m})}\bigg], \;1 \!\leq\! m \!\leq\! M, \\
&[\mathbf{I}(\QV)]_{2,m+4} \!=\! \frac{(\delta^y_m-\delta^y)(r_m- \lambda_m)}{2\pi \upsilon_m^4}\bigg[\frac{e^{-\frac{(r_m - \lambda_m)^2}{\upsilon^2_m}}}{\Phi(\frac{r_m-\lambda_m}{\upsilon_m})}\nonumber\\
&\;\;\;\;\;\;\;\;\;\;\;\;\;\;\;\;\;\;\;\; \!+\!\frac{e^{-\frac{(r_m - \lambda_m)^2}{\upsilon^2_m}}}{\Phi(\frac{-r_m+\lambda_m}{\upsilon_m})}\bigg], \;1 \!\leq\! m \!\leq\! M, \\
&[\mathbf{I}(\QV)]_{3,m+4} \!=\! \frac{(\delta^z_m-\delta^z)(r_m- \lambda_m)}{2\pi \upsilon_m^4}\bigg[\frac{e^{-\frac{(r_m - \lambda_m)^2}{\upsilon^2_m}}}{\Phi(\frac{r_m-\lambda_m}{\upsilon_m})}\nonumber\\
&\;\;\;\;\;\;\;\;\;\;\;\;\;\;\;\;\;\;\;\; \!+\!\frac{e^{-\frac{(r_m - \lambda_m)^2}{\upsilon^2_m}}}{\Phi(\frac{-r_m+\lambda_m}{\upsilon_m})}\bigg], \;1 \!\leq\! m \!\leq\! M, \\
&[\mathbf{I}(\QV)]_{4,m+4} \!=\! \frac{(r_m- \lambda_m)}{2\pi \upsilon_m^4}\bigg[\frac{e^{-\frac{(r_m - \lambda_m)^2}{\upsilon^2_m}}}{\Phi(\frac{r_m-\lambda_m}{\upsilon_m})}\nonumber\\
&\;\;\;\;\;\;\;\;\;\;\;\;\;\;\;\;\;\;\;\; \!+\!\frac{e^{-\frac{(r_m - \lambda_m)^2}{\upsilon^2_m}}}{\Phi(\frac{-r_m+\lambda_m}{\upsilon_m})}\bigg], \;1 \!\leq\! m \!\leq\! M.
\end{align}
\section{Numerical Experiments}
\label{sec:numexp}
We investigated the performance of our proposed method through numerical simulations. We also compared the performance of one-bit processing with full precision measurements. We used MATLAB CVX package to solve optimizations in \eqref{eq-op-re-con} and \eqref{eq-opt-sdp} \cite{cvx}. All the experiments are conducted under identical conditions
under Matlab R2018a on a PC equipped with an operating
system of Windows 10 64-bit, an Intel i7-6820HQ 2.70GHz
CPU, and a 8GB RAM. Throughout all the experiments, we define signal-to-noise ratio (SNR) (in dB) at the $m$-th node as
\par\noindent\small
$$\mathrm{SNR}_m = 10\log_{10}\frac{|\alpha_m|^2 \|\SV(\tau_m)\|^2}{\sigma_m^2}.$$
\normalsize

\noindent\textbf{One-bit time-delay estimation}: For $100$ digital samples obtained at the Nyquist rate, i.e. $L=100$ and $\vartheta=1$,  Fig.~\ref{fig:NRMSEvsSNR} shows the normalized root-mean-squared-error (N-RMSE) of the time-delay estimates, computed over $1000$ Monte Carlo trials,
with respect to SNR. This \textit{estimation N-RMSE} is $\frac{\sqrt{\sum_{j=1}^J(\widehat{\tau}_{m,j}-\tau_m)^2}}{\tau_m J}$ where $\widehat{\tau}_{m,j}$
denotes the time-delay estimate at the $j$-th Monte Carlo trial and $J$ is the number of Monte Carlo trials. We assume $s(t)$ to be $\pi/2$-BPSK-modulated with a raised cosine shaping filter of the bandwidth $180$ KHz and the roll-off factor $1$. The temporal threshold $\GAMMA_m$ is randomly drawn from a uniform distribution with support $[-A_{\rm max},A_{\rm max}]$, where $A_{\rm max}$ denotes the maximum amplitude of the received signal at NB-IoT nodes. We observe that to achieve the same N-RMSE, the SNR should be about $5$ dB higher for one-bit processing than the full-precision case.\\

\noindent\textbf{Effect of oversampling}: As discussed in Section \ref{sec:oversampling}, oversampling compensates the performance loss arising from the one-bit quantization scheme. Fig.~\ref{fig:NRMSEvsK} shows the N-RMSE of the time-delay estimates versus the oversampling factor, i.e., $\vartheta$, at ${\rm SNR}=-5$ dB. As predicted in theory, the N-RMSE of oversampled one-bit processing with $\vartheta=5$ approaches that of the full-precision processing.\\
\begin{figure}[t]
\centering
\includegraphics[width=0.72\columnwidth]{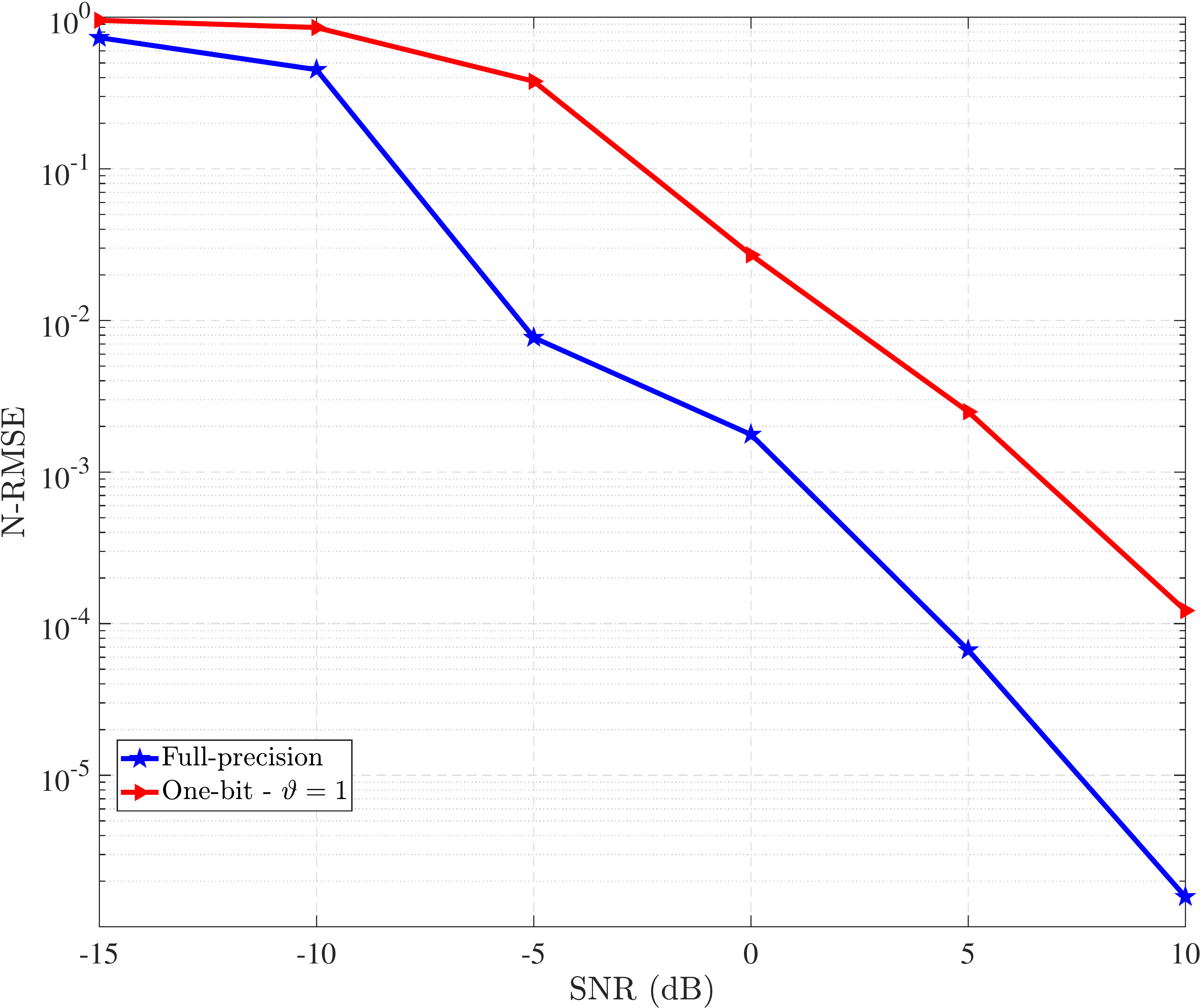}
\DeclareGraphicsExtensions.
\caption{N-RMSE of the time-delay estimates versus the SNR with $L=100$ and $\vartheta=1$. The signal $s(t)$ is a $\pi/2$-BPSK modulated signal with bandwidth $B=180$ KHz.}
\label{fig:NRMSEvsSNR}
\end{figure}
\begin{figure}[t]
\centering
\includegraphics[width=0.72\columnwidth]{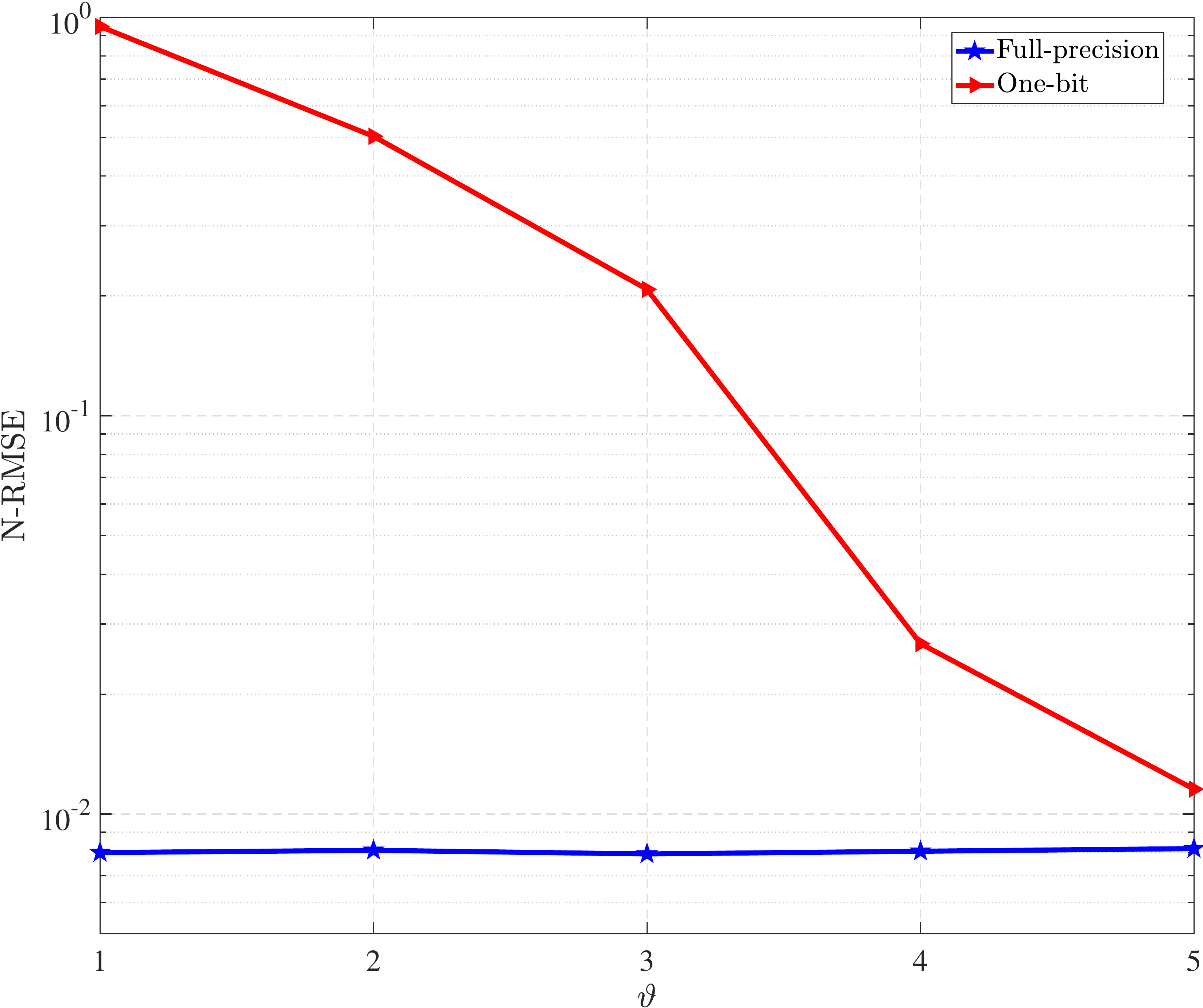}
\DeclareGraphicsExtensions.
\caption{N-RMSE of the time-delay estimates versus the the oversampling factor $\vartheta$ with $L=100$ and ${\rm SNR}=-5$ dB. The signal $s(t)$ is a $\pi/2$-BPSK modulated signal with bandwidth $B=180$ KHz.}
\label{fig:NRMSEvsK}
\end{figure}

\noindent\textbf{Localization with different node geometries}: Next, we investigate our proposed localization method for various node placements. We consider three node geometries: uniform circular (Fig.~\ref{fig:cirshape}), uniform linearly-spaced in an L-shape (Fig.~\ref{fig:Lshap}), and random (Fig.~\ref{fig:Ranshape}). To show the performance over different ranges, we consider the performance of these geometries over small ($[-800~{\rm m}, 800~{\rm m}] \times [-800~{\rm m}, 800~{\rm m}]$), large ($[-2000~{\rm m}, 2000~{\rm m}] \times [-2000~{\rm m}, 2000~{\rm m}]$), and mid-size ($[-1200~{\rm m}, 1200~{\rm m}] \times [-1200~{\rm m}, 1200~{\rm m}]$) areas, respectively. In Fig.~\ref{fig:cirshape}, the nodes were spaced on a circle with radius of $800$ m and the target and the base-station were randomly placed at $[-309~{\rm m}, 287~{\rm m}]$ and $[-208,~{\rm m}, -312~{\rm m}]$ (in $X$-$Y$ Cartesian coordinate system), respectively.
When the nodes were configured in L-shape and randomly, the target was randomly placed at $[371.7~{\rm m}, -338.4~{\rm m}]$ and $[-615.8~{\rm m}, -753.8~{\rm m}]$ and the base station was randomly located at $[-98~{\rm m}, 1112~{\rm m}]$ and $[-87~{\rm m}, 53~{\rm m}]$, respectively.

To consider the impact of the relative distances of the different nodes to the target of interest on the SNR, we generate the SNR at the $m$-th node ($m>1$) as ${\rm SNR}_m = {\rm SNR}_1 (\frac{d_m}{d_1})^2$ where ${\rm SNR}_1$ denotes the SNR at the reference node, which is assumed to be $0$ dB in Figs.~\ref{fig:cirshape}, \ref{fig:Lshap}, and \ref{fig:Ranshape}. The temporal thresholds and $s(t)$ are generated similar to Fig. \ref{fig:NRMSEvsSNR}. The maximum detectable range by NB-IoT nodes, i.e., $r_{\rm max}$, was considered to be $4000$ m. The positive thresholds $\lambda_m$'s were randomly drawn from $8$ predetermined values over the interval $(0, r_{\rm max}]$. These thresholds are encoded with $3$ bits and transmitted to the FC along with one-bit range information. 

Our ANTARES algorithm estimates the target location with errors of $22.89$, $23.87$, and $21.52$ m for circular, L-shape, and random geometries, respectively. This is very close to that of the optimal method given in Theorem \ref{theorem-2}, wherein the corresponding errors are $6$, $9.4$, and $7.81$ m, respectively; the errors in the full-precision methods are $1$ m, $1.2$, and $1.06$ m, respectively. This indicates the robustness of our method against distribution in of NB-IoT nodes. In order to draw a comparison between the computational complexities of ANTARES and the optimal method, we take account of their corresponding run-times for the investigated scenarios in Figs.~\ref{fig:cirshape}, \ref{fig:Lshap}, and \ref{fig:Ranshape}, which are, respectively, $3.27$ s, $3.63$ s, and $3.91$ s for ANTARES besides $81.39$ s, $88.53$ s, and $85.74$ s for the optimal method. This implies that ANTARES is considerably more computationally efficient than the optimal method in Theorem \ref{theorem-2}.
\begin{figure}[t]
\centering
\includegraphics[width=0.72\columnwidth]{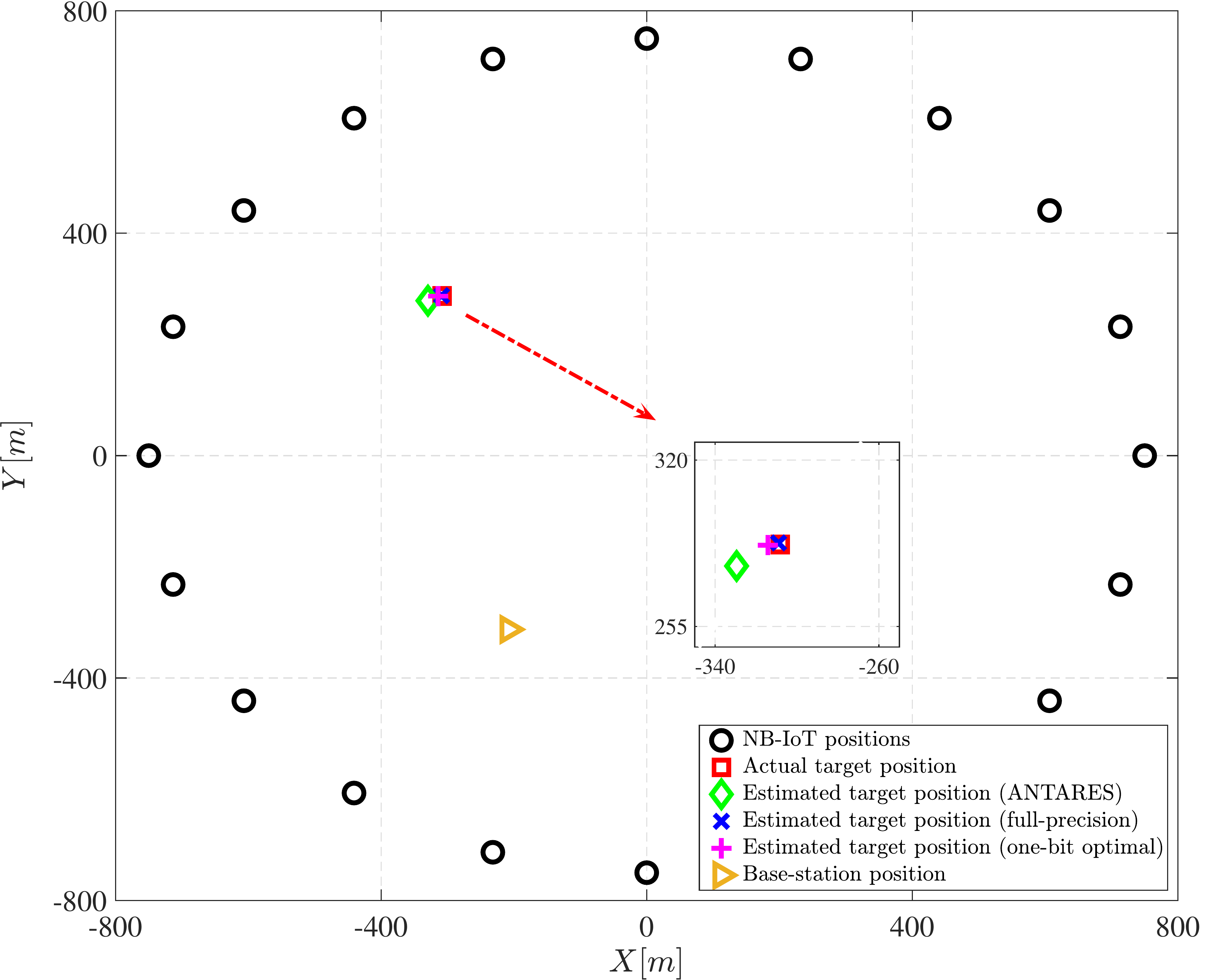}
\DeclareGraphicsExtensions.
\caption{Localization with $M=20$ NB-IoT nodes (black circles) uniformly spaced on a circle with radius of $800$ m. The target-of-interest is randomly placed at $(-309~{\rm m}, 287~{\rm m})$. The SNR at all the NB-IoT nodes is $0$ dB. 
}
\label{fig:cirshape}
\end{figure}
\begin{figure}[t]
\centering
\includegraphics[width=0.72\columnwidth]{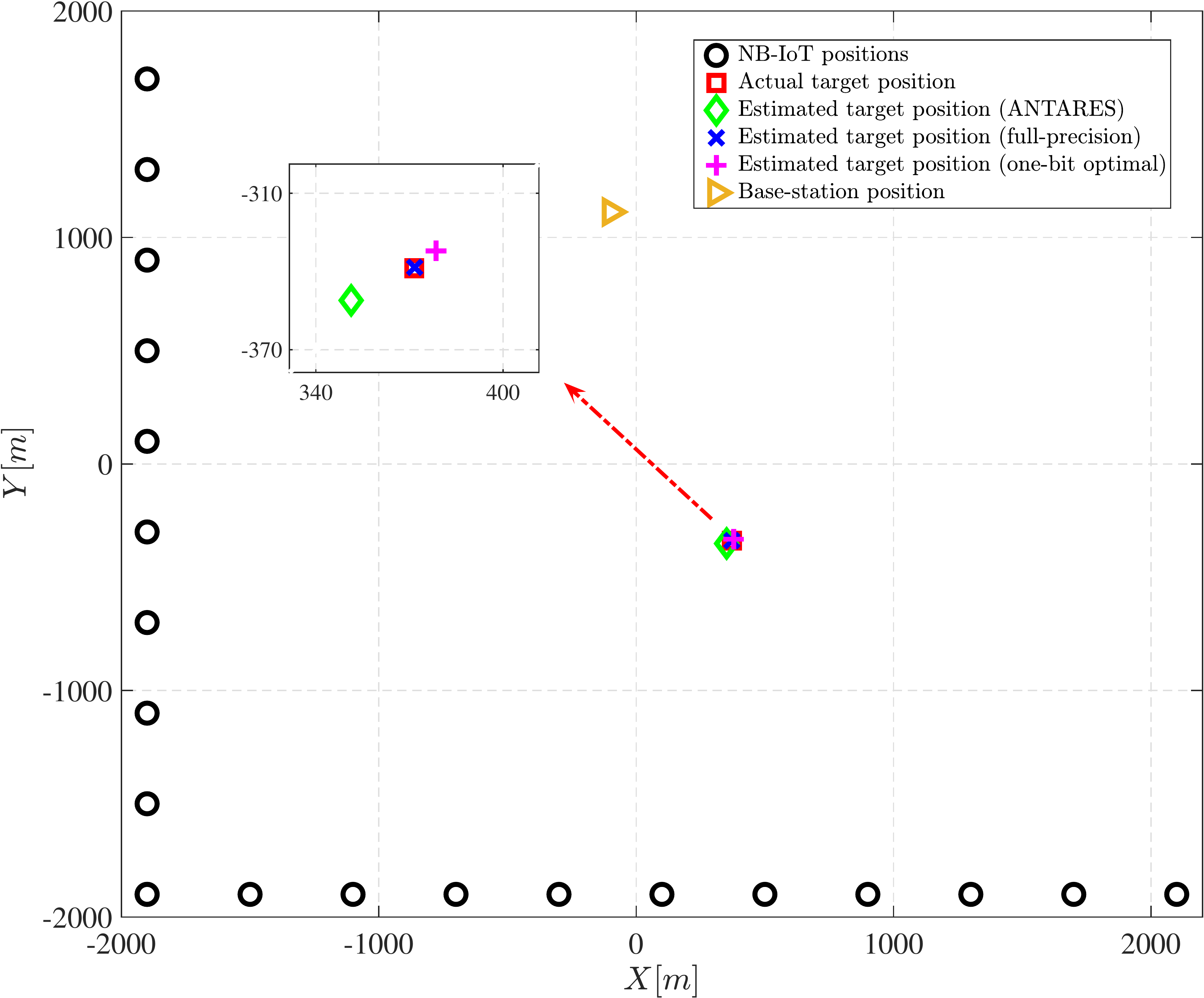}
\DeclareGraphicsExtensions.
\caption{Localization with $M=20$ NB-IoT nodes (black circles) linearly spaced in an L-shape. The target-of-interest is randomly placed at $(371~{\rm m}, -338~{\rm m})$. The SNR at the $m$-th node ($m>1$) is ${\rm SNR}_m = {\rm SNR}_1 \left(\frac{d_m}{d_1}\right)^2$ with ${\rm SNR}_1=0$ dB.}
\label{fig:Lshap}
\end{figure}
\begin{figure}[t]
\centering
\includegraphics[width=0.72\columnwidth]{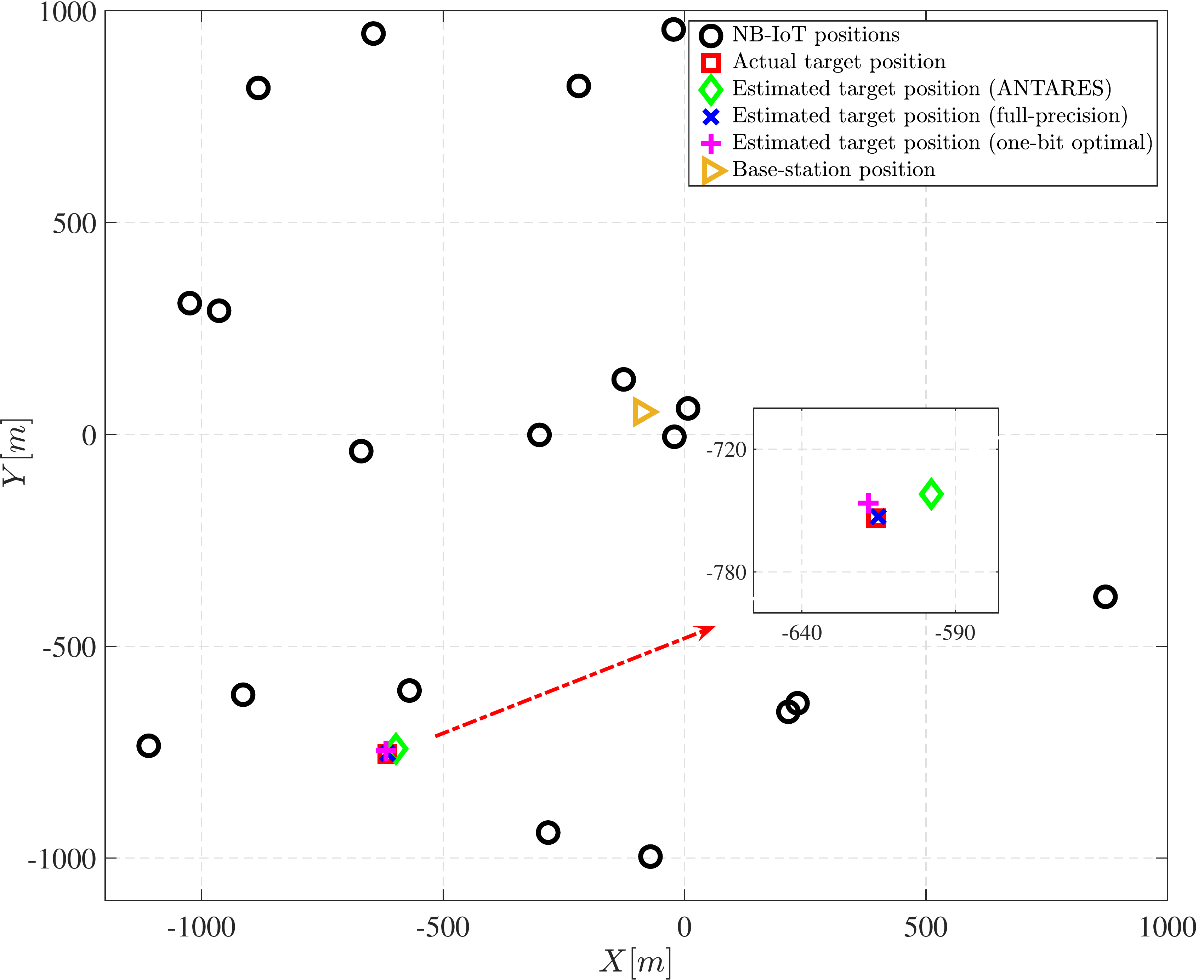}
\DeclareGraphicsExtensions.
\caption{Localization with $M=20$ NB-IoT nodes (black circles) randomly distributed over the area $[-1200~{\rm m}, 1200~{\rm m}] \times [-1200~{\rm m}, 1200~{\rm m}]$. The target-of-interest is randomly placed at $(1160~{\rm m}, -340~{\rm m})$. The SNR at the $m$-th node ($m>1$) is ${\rm SNR}_m = {\rm SNR}_1 \left(\frac{d_m}{d_1}\right)^2$ with ${\rm SNR}_1=0$ dB.}
\label{fig:Ranshape}
\end{figure}

Next, for the random geometry, we show the effect of decreasing ${\rm SNR}_1$ to $-5$ dB (Fig.~\ref{fig:RanshapeSNR0}). The error with ANTARES algorithm now degrades to $59.85$ m compared to $12.4$ and $3.4$ m observed in the optimal and full-precision approaches.
\\
\begin{figure}[t]
\centering
\includegraphics[width=0.72\columnwidth]{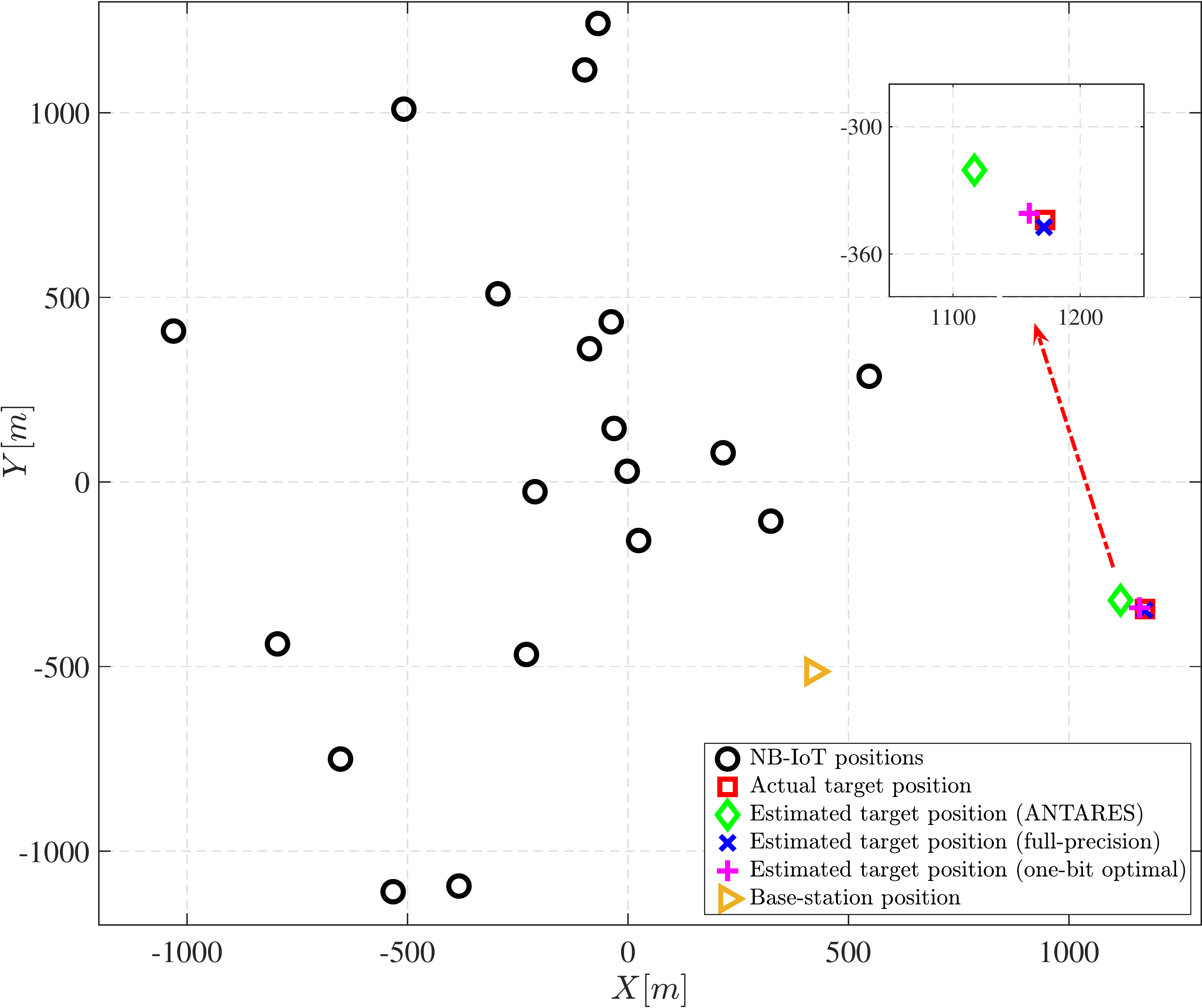}
\DeclareGraphicsExtensions.
\caption{Localization with $M=20$ NB-IoT nodes (black circles) randomly distributed within the area $[-1200~{\rm m}, 1200~{\rm m}] \times [-1200~{\rm m}, 1200~{\rm m}]$. The target-of -nterest is randomly placed at $(-618~{\rm m}, -338~{\rm m})$. The SNR at the $m$-th node ($m>1$) is ${\rm SNR}_m = {\rm SNR}_1 \left(\frac{d_m}{d_1}\right)^2$ with ${\rm SNR}_1=-5$ dB.}
\label{fig:RanshapeSNR0}
\end{figure}

\noindent\textbf{Statistical performance}: Figs.~\ref{fig:mseVsM-a} illustrates the \textit{localization N-RMSE}, i.e. N-RMSE in the estimation of the target location, with respect to the number NB-IoT nodes $M$,  defined as
$
\frac{\sqrt{\sum\limits_{j=1}^J (\delta^x-\widehat{\delta}^x_j)^2 + (\delta^y-\widehat{\delta}^y_j)^2 }}{J\sqrt{ {\delta^x}^2 + {\delta^y}^2 }}$, 
where $[\widehat{\delta}^x_j, \widehat{\delta}^y_j ]^T$
denotes the target location estimate at the $j$-th Monte Carlo trial and $J$ is the number of Monte Carlo trials.
Figs.~\ref{fig:mseVsM-a} plots the normalized-root-localization-CRB, i.e., $\sqrt{\frac{[\mathbf{I}^{-1}(\QV)]_{1,1}+[\mathbf{I}^{-1}(\QV)]_{2,2}}{ {\delta^x}^2 + {\delta^y}^2 }}$ where $\mathbf{I}(\QV)$ is specified in Section \ref{sec:crb}.
The nodes and targets were placed randomly over ($[-800~{\rm m}, 800~{\rm m}] \times [-800~{\rm m}, 800~{\rm m}]$) area during each of the $200$ Monte Carlo trials. The SNR at the $m$-th node ($m>1$) is assumed to be ${\rm SNR}_m = {\rm SNR}_1 \left(\frac{d_m}{d_1}\right)^2$ with ${\rm SNR}_1= -2$ dB. Further, the temporal thresholds, $s(t)$ and $\lambda_m$'s are generated similar to Figs. \ref{fig:NRMSEvsSNR} and \ref{fig:Lshap}. We observe that the N-RMSEs of the proposed optimal and ANTARES methods improve with increase in $M$. The N-RMSE for the optimal method is very close to the normalized root of the CRB and it approaches to that of the full-precision when $M>80$. It is also seen that the normalized CRB tends to the N-RMSEs of the full-precision at the high number of sensors. In addition, Fig~\ref{fig:mseVsM-b} shows the relative N-RMSE, namely the difference in N-RMSE of the optimal and ANTARES methods as well as the normalized CRB relative to that of full-precision. We observe that the relative N-RMSE rises by $2.2$\%, $0.6$\% and $0.3$\% in case of ANTARES, optimal methods and the CRB, respectively, over the full-precision approach when $M=20$. The observed difference in the estimation performance of ANTARES and optimal approaches arises from the fact that the alternating approach employed for ANTARES is guaranteed to converge to only a local minimum of the optimization problem in (44) \cite{bezdek2003convergence}, while the optimal method always provides the global minimum of (44).

The temporal thresholds were randomly generated in all experiments. Comparing the localization accuracy in Figs. \ref{fig:cirshape}-\ref{Fig5} show that variations in temporal thresholds do not have considerable influence on the overall localization performance.
\begin{figure}[!t]
\centering
\subfloat[]{\includegraphics[width=0.72\columnwidth]{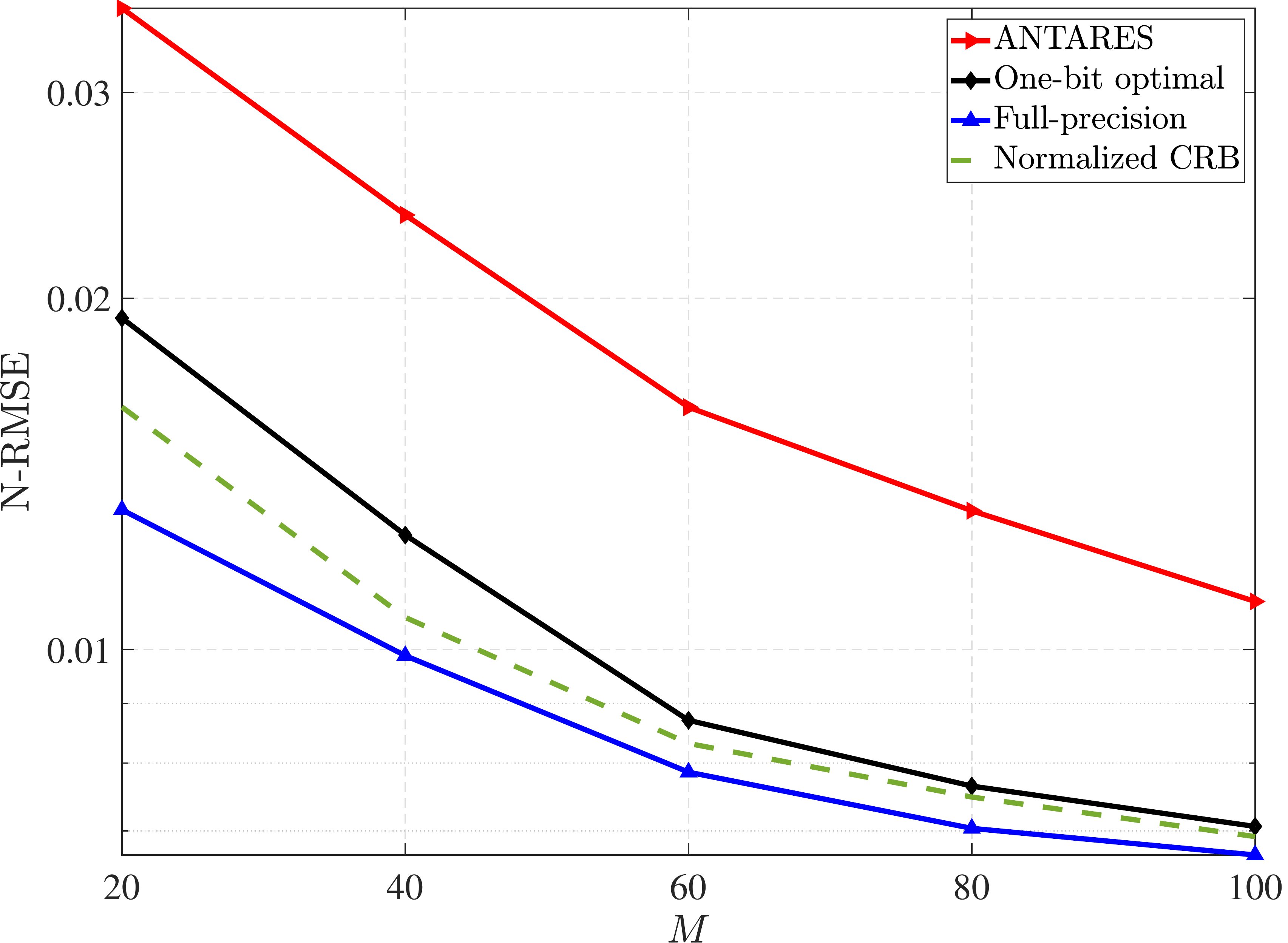}%
\label{fig:mseVsM-a}}
\\
\subfloat[]{\includegraphics[width=0.72\columnwidth]{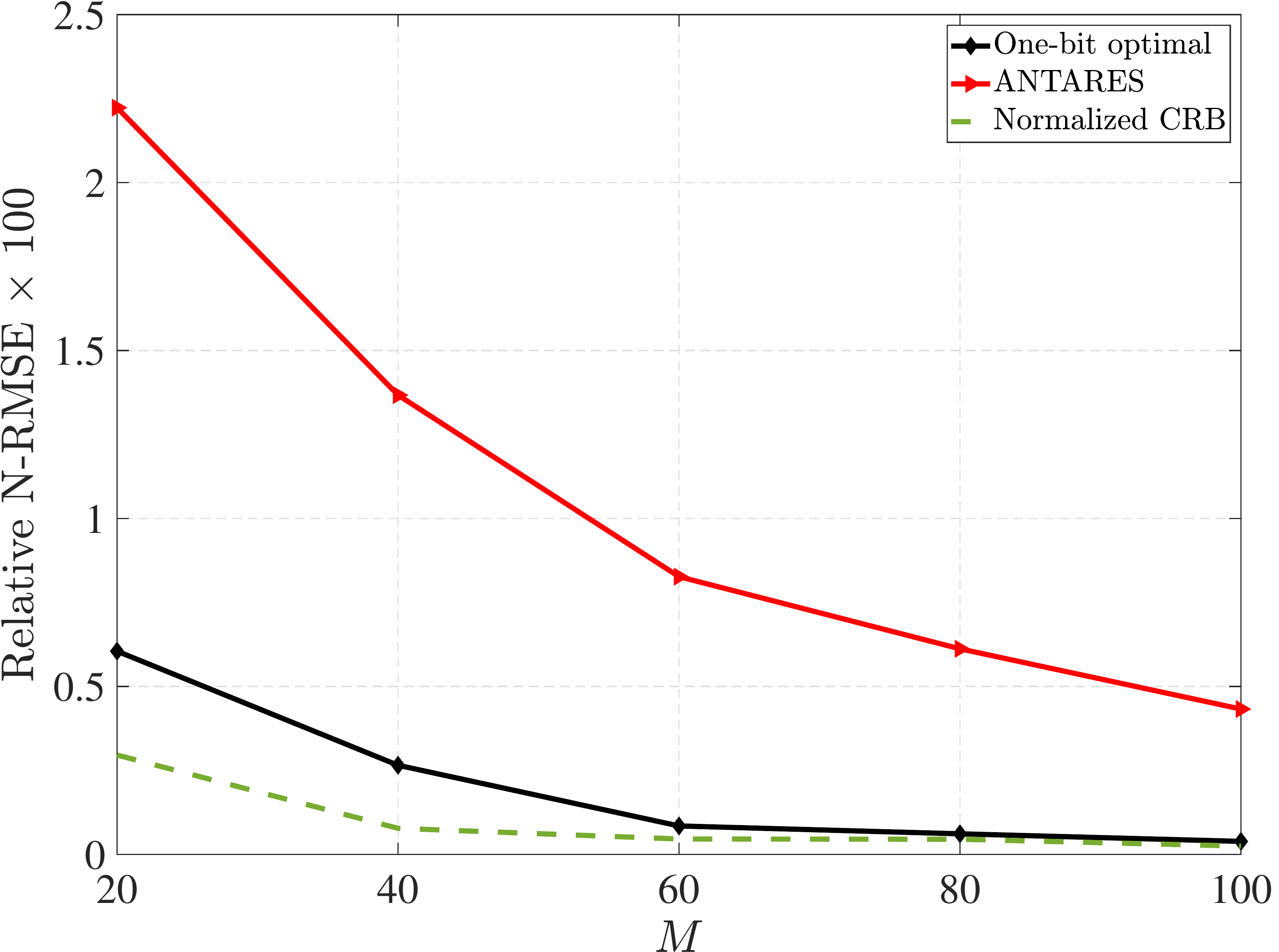}%
\label{fig:mseVsM-b}}
\caption{(a) N-RMSE and (b) Relative N-RMSE in the estimated target location with respect to the number of IoT devices $M$. The SNR at the $m$-th node ($m>1$) is ${\rm SNR}_m = {\rm SNR}_1 \left(\frac{d_m}{d_1}\right)^2$ with ${\rm SNR}_1= -2$ dB.}
\label{Fig5}
\end{figure}
%
\section{Summary}
\label{sec:summ}
In summary, the one-bit sampling offers an attractive solution to the challenges posed by the NB-IoT for location-based services. The one-bit samplers are integral to developing low cost and low power devices. We proposed a one-bit passive sensor array formulation to estimate the time-of-arrival in an NB-IoT network. The quantized samples of the estimates are then forwarded to an FC. We propose a novel method that casts the localization problem from aggregated quantized nodal estimates as a multivariate fractional optimization problem that we solve using the optimal Lasserre's SDP relaxation. We also propose the ANTARES algorithm as an alternative sub-optimal method with reduced computational complexity compared to Lasserre's. Our approach is helpful in addressing the problem of maintaining high localization accuracy while deploying reduced-rate ADCs at the nodes as well as limited-capacity NB-IoT links.
\appendices

\section{Proof of Lemma \ref{lem-1-2}}
\label{App-A-A}
The optimization problems \eqref{eq-optimization} and \eqref{eq-op-re-con} are equivalent. Hence, it suffices to prove this for only \eqref{eq-optimization}. 
Define
\begin{align}
 \mathbf{y}_m^\circ = \widetilde{\alpha}_m^\circ \SV(\widetilde{\tau}_m^\circ) + \alpha_m^\circ \SV(\tau_m^\circ) + \NV_m^\circ,   
\end{align}
where $[\widetilde{\alpha}_m^\circ, \alpha_m^\circ, \widetilde{\tau}_m^\circ, \tau_m^\circ, \NV_m^\circ]^T \neq [\widetilde{\alpha}_m, \alpha_m, \widetilde{\tau}_m, \tau_m, \NV_m]^T$ and thus, $\mathbf{y}_m^\circ \neq \mathbf{y}_m$.
It suffices to show that $[Q(\YV_m^\circ)]_l \neq [Q(\YV_m)]_l$ at least for one $l$ as $L \to \infty$. The previous statement holds only if, at least for one $l$, the following occurs:
\begin{align}\label{eq-App-A-A-1}
\left\{\begin{array}{ll}
\RE\{[\YV_m]_l\} > \RE\{[\GAMMA_m]_l\} > \RE\{[\YV_m^\circ]_l\}, ~~{\rm or}, \\
\RE\{[\YV_m]_l\} < \RE\{[\GAMMA_m]_l\} < \RE\{[\YV_m^\circ]_l\},  ~~{\rm or}, \\
\IM\{[\YV_m]_l\} > \IM\{[\GAMMA_m]_l\} > \IM\{[\YV_m^\circ]_l\}, ~~{\rm or}, \\
\IM\{[\YV_m]_l\} < \IM\{[\GAMMA_m]_l\} < \IM\{[\YV_m^\circ]_l\}.
\end{array}\right.
\end{align}
Let ${\cal A}$ denote the event described by \eqref{eq-App-A-A-1} for a given $l$. In practice, the real and imaginary parts of $[\YV_m]_l$ and $[\YV_m^\circ]_l$ are upper bounded by, say, $A_{\rm max}$.
Then, probability of ${\cal A}$ is \cite{gianelli2016one}
\par\noindent\small
\begin{align}
{\rm Pr}({\cal A}) = \frac{|\RE\{[\YV_m]_l\} -\RE\{[\YV_m^\circ]_l\}|}{2A_{\rm max}} + \frac{|\IM\{[\YV_m]_l\} -\IM\{[\YV_m^\circ]_l\}|}{2A_{\rm max}}.  
\end{align}
\normalsize
The probability that \eqref{eq-App-A-A-1} occurs at least for one $l$, denoted by ${\cal H}$, is
\par\noindent\small
\begin{align}
 {\rm Pr}({\cal H}) = 1 - \prod_{l=1}^{L}& \bigg(1 -  \frac{|\RE\{[\YV_m]_l\} -\RE\{[\YV_m^\circ]_l\}|}{2A_{\rm max}} \nonumber\\
 &- \frac{|\IM\{[\YV_m]_l\} -\IM\{[\YV_m^\circ]_l\}|}{2A_{\rm max}}\bigg).
\end{align}
\normalsize
From \cite{gianelli2016one}, $1-x \leq e^{-x}, \forall x \in \mathds{R}$. Hence, it follows that
\par\noindent\small
\begin{align}
 {\rm Pr}({\cal H}) \geq 1 - e^{-\sum_{l=1}^{L}\frac{|\RE\{[\YV_m]_l\} -\RE\{[\YV_m^\circ]_l\}|}{2A_{\rm max}} - \frac{|\IM\{[\YV_m]_l\} +\IM\{[\YV_m^\circ]_l\}|}{2A_{\rm max}}}.
\end{align}
\normalsize
But $\YV_m^\circ \neq \YV_m$. Thus, $-\sum_{l=1}^{L}\frac{|\RE\{[\YV_m]_l\} -\RE\{[\YV_m^\circ]_l\}|}{2A_{\rm max}} - \frac{|\IM\{[\YV_m]_l\} +\IM\{[\YV_m^\circ]_l\}|}{2A_{\rm max}} \to \infty$ as $L \to \infty$, and
    $\lim_{L \to \infty }  {\rm Pr}({\cal H}) = 1$. 
This implies that $\YV_m$ is the only point which satisfies the constraints in \eqref{eq-optimization} as $L \to \infty$. Accordingly, as $L \to \infty$, the optimization problem \eqref{eq-optimization} reduces to the LASSO estimator which has been shown to be consistent \cite{chatterjee2011strong}. This completes the proof.
\section{Proof of Theorem\ref{theorem-1}}
\label{App-A}
To show that \eqref{eq:jsdp} is equivalent to \eqref{eq:fracopt}, we first prove that the global minimum of \eqref{eq:fracopt} coincides with that of \eqref{eq:jsdp}. Assume that $\RV^\star_{o}$ and $\begin{bmatrix} \RV^{\star T}_{e} & v^\star \end{bmatrix}^T$ are the minimizers of \eqref{eq:fracopt} and \eqref{eq:jsdp}, respectively.
Define a set
$
\KSET = \left\{ \RV \in \mathds{R}_{ \geq 0}^{M} ~|~  \WV \odot (\RV - \LAMBDA) \succeq \ZEROVV \right\}.
$
Given ${\cal J}(\RV) \geq 0$ for $\RV \in \KSET$, it readily follows from the first constraint in \eqref{eq:jsdp} that $\dfrac{{\cal F}(\RV^\star_{e})}{{\cal J}(\RV^\star_{e})} \leq v^\star$. Considering that $\RV^\star_{e}$ belongs to the feasible set of \eqref{eq:fracopt}, i.e., $\RV^\star_{e} \in \KSET$, we obtain
\par\noindent\small
\begin{align}
\label{eq-27}
\dfrac{{\cal F}(\RV^\star_{o})}{{\cal J}(\RV^\star_{o})} \leq \dfrac{{\cal F}(\RV^\star_{e})}{{\cal J}(\RV^\star_{e})} \leq v^\star.
\end{align}
\normalsize
On the other hand, defining $v_o = \frac{{\cal F}(\RV^\star_{o})}{{\cal J}(\RV^\star_{o})}$ and considering $\RV^\star_o \in \KSET$, it follows that $\begin{bmatrix} \RV^{\star T}_{o} & v_o\end{bmatrix}^T$ is in the feasible set of \eqref{eq:jsdp}. Therefore, \par\noindent\small
\begin{align}
    \label{eq-28}
v^\star \leq v_{o} =  \dfrac{{\cal F}(\RV^\star_{o})}{{\cal J}(\RV^\star_{o})},
\end{align}
\normalsize
Now, comparing \eqref{eq-27} and \eqref{eq-28} implies that \eqref{eq:fracopt} and \eqref{eq:jsdp} share the same global minimum, i.e.,
\par\noindent\small
\begin{align}
\label{eq-29}
 v^\star =  \dfrac{{\cal F}(\RV^\star_{o})}{{\cal J}(\RV^\star_{o})} .
\end{align}
\normalsize
Further deduction from \eqref{eq-27} and \eqref{eq-29} yields
\par\noindent\small
\begin{align}
  \dfrac{{\cal F}(\RV^\star_{o})}{{\cal J}(\RV^\star_{o})} = \dfrac{{\cal F}(\RV^\star_{e})}{{\cal J}(\RV^\star_{e})},
\end{align}
\normalsize
indicating $\RV^\star_e$ is also a minimizer of \eqref{eq:fracopt}. This completes the proof.
\section{Proof of Theorem \ref{theorem-2}}
\label{App-B}
\subsection{Preliminaries to the Proof}
\label{App-B-1}
Recall the definition of sum-of-squares (SOS) polynomial and a useful related result as follows.
\begin{Def}[Sum-of-squares]
\label{definition-3}
A polynomial ${\cal P}(\UV)$ of degree $2q$ is sum-of-squares (SOS) if and only if there exist polynomials ${\cal Y}_1(\UV), \cdots, {\cal Y}_I(\UV)$ of degree $q$ such that ${\cal P}(\UV) = \sum_{i=1}^I {\cal Y}^2_i(\UV)$.
\end{Def}
\begin{lem}
\label{lemma-1}
Given $\mathds{P}$ as the set of SOS polynomials and polynomials ${\cal E}_i(\UV)$ for $1 \leq i \leq I$, define the sets
\par\noindent\small
\begin{align}
\label{eq-compset}
 \mathds{W} &\!=\!\{\UV \in \mathds{R}^n \mid {\cal E}_i(\UV) \geq 0, \forall i \in \{1,2 ,\cdots, I\} \}  \\
\label{eq-positset}
\mathds{G}_p & \!=\! \Scale[0.86]{\bigg\{ \sum_{i=0}^I {\cal P}_i(\UV) {\cal E}_i(\UV) \mid {\cal E}_0(\UV) \!=\! 1, 
{\cal P}_i(\UV) \! \in \! \mathds{P}, \deg\left({\cal P}_i(\UV) {\cal E}_i(\UV)\right) \! \leq \! 2p \bigg\},}
\end{align}
\normalsize
such that $\mathds{W}$ is compact and there exists a polynomial ${\cal U}(\UV) \in \mathds{G}_p$ where $\{\UV \in \mathds{R}^n \mid {\cal U}(\UV) \geq 0\}$ is compact. Then, a polynomial ${\cal B}(\UV)$ of degree $q$ is strictly positive on $\mathds{W}$, i.e., ${\cal B}(\UV) >0$ $\forall \UV \in \mathds{W}$, if and only if ${\cal B}(\UV) \in \mathds{G}_p$ for some integer $p \geq \max \left(\lceil q \rceil, \underset{i}{\max} \left\lceil \frac{\deg({\cal E}_i)}{2} \right\rceil \right)$.
\end{lem}
\begin{IEEEproof}
We refer the reader to \cite{nie2007complexity}.
\end{IEEEproof}
\subsection{Proof of the Theorem}
\label{App-B-2}
We first show that \eqref{eq:jsdp} satisfies the conditions stated in Lemma~\ref{lemma-1} of Appendix~\ref{App-B-1}. In consequence, it can be reformulated as minimization of a positive polynomial function on a compact set. Lasserre has shown that minimizer of a positive polynomial function on a compact set can be obtained through solving an equivalent SDP \cite[Theorem 4.2]{lasserre2001global}. Thus, we ultimately resort to \cite[Theorem 4.2]{lasserre2001global} to recast the resulting optimization problem as an SDP. 

Consider ${\cal E}_i$'s to be the inequality constraints of \eqref{eq:jsdp}. Then, we need to prove the following three statements:  
\begin{enumerate} 
\item The feasible set of \eqref{eq:jsdp} is compact.
\item A polynomial ${\cal U}([\RV,v]^T) \in \mathds{G}_p$ exists such that $\{\RV \in \mathds{R}^M, v \in \mathds{R} \mid {\cal U}([\RV, v]^T) \geq 0\}$ is compact.
\item The objective function of \eqref{eq:jsdp} is strictly positive on its feasible set.
\end{enumerate}

For the first statement, note that the feasible set contains all of its boundary points and is therefore closed. From Heine-Borel Theorem \cite{rudin1964principles}, to show compactness of the feasible set, it suffices to show that it is bounded. To this end, note the constraint on the value of $\RV$ which is limited by the maximum detectable range $r_{\rm max} \in \mathds{R}_{> 0}$ of the NB-IoT nodes 
so that $r_m \leq r_{\rm max}$ for all $m \in \mathds{M}$. 
This implies that the continuous function $\frac{{\cal F}(\RV)}{{\cal J}(\RV)}$ is bounded on $\mathds{T} = \{\RV \in \mathds{R}^M \mid r_m \leq r_{\rm max},\; \forall m \in \mathds{M}\}$ \cite[Theorem 4.16]{rudin1964principles}. In other words, $\frac{{\cal F}(\RV)}{{\cal J}(\RV)} \leq \varphi$, where $\varphi = \underset{\RV \in \mathds{T}}{\textrm{maximize}} ~~ \frac{{\cal F}(\RV)}{{\cal J}(\RV)}$. The optimization problem in \eqref{eq:jsdp} is indeed a minimization of an upper bound of $\frac{{\cal F}(\RV)}{{\cal J}(\RV)}$, i.e. $v$. Without loss of generality, assume $v \leq v_{\rm max}$ where $v_{\rm max} \geq \varphi$. These practical constraints on $\RV$ and $v$ do not change the solution of \eqref{eq:jsdp} but guarantee boundedness and thereby compactness of the its feasible set. On the other hand, it is possible to show the boundedness of $v$, in turn, entails the boundedness of $\RV$.
To show that, 
let assume $\mathds{B}$ to be an arbitrary subset of $\{1,\cdots,M\}$ and define $\CV$ such that $[\CV]_k = [\RV]_k$ for $k \in \mathds{B}$.
When $v \leq v_{\rm max}$, from \eqref{eq-21} and \eqref{eq-21-1}, we get
\par\noindent\small
\begin{align}
\label{eq-limitconst}
&\lim_{\CV \to \infty} v {\cal J}(\RV) - {\cal F}(\RV) = \nonumber\\ & - \dfrac{1}{4} \lim_{\CV \to \infty} \bigg( \|\Pi^{\perp}_{\VM} (\overline{\RV}-r_1 \ONEVV)\|^2_2 \| \Pi^{\perp}_{\VM} \big[ (\overline{\RV}-r_1 \ONEVV) \odot (\overline{\RV}-r_1 \ONEVV) \big]\|^2_2 \nonumber\\
& - \left( \big[ (\overline{\RV}-r_1 \ONEVV) \odot (\overline{\RV}-r_1 \ONEVV) \big]^T \Pi^{\perp}_{\VM} (\overline{\RV}-r_1 \ONEVV) \right)^2 \bigg),     
\end{align}
\normalsize
Using Cauchy–Schwarz inequality and idempotency of $\Pi^{\perp}_{\VM}$, we have
\par\noindent\small
\begin{align}
\label{eq-CauSchineq}
\|\Pi^{\perp}_{\VM} (\overline{\RV}-r_1 \ONEVV)\|^2_2 \| \Pi^{\perp}_{\VM} \big[ (\overline{\RV}-r_1 \ONEVV) \odot (\overline{\RV}-r_1 \ONEVV) \big]\|^2_2 \geq \nonumber\\
 \left( \big[ (\overline{\RV}-r_1 \ONEVV) \odot (\overline{\RV}-r_1 \ONEVV) \big]^T \Pi^{\perp}_{\VM} (\overline{\RV}-r_1 \ONEVV) \right)^2.   
\end{align}
\normalsize
It follows from \eqref{eq-limitconst} and \eqref{eq-CauSchineq} that, when $v \leq v_{\rm max}$ and as each $r_m$ approaches infinity, the constraint $v {\cal J}(\RV) - {\cal F}(\RV)$ becomes negative. Hence, when $v \leq v_{\rm max}$, to ensure $v {\cal J}(\RV) - {\cal F}(\RV) \geq 0$, the ranges $r_m$, $m \in \mathds{M}$ must be bounded. This implies that $v \leq v_{\rm max}$ is sufficient for the compactness of the feasible set of \eqref{eq:jsdp}. Accordingly, without loss of generality, the optimization problem \eqref{eq:jsdp} becomes
\par\noindent\small
\begin{align}
\label{eq-opt-compact}
\begin{array}{ll}
\underset{v, \RV}{\textrm{minimize}} & v \\
\SUBJTO &  v {\cal J}(\RV) -{\cal F}(\RV) \geq 0, \\
  & \WV \odot
(\RV - \LAMBDA) \succeq \ZEROVV, \\
  & \RV \succeq \ZEROVV,\\
  & v_{\rm max} - v \geq 0,
\end{array}
\end{align}
\normalsize
in which the feasible set is compact. Note that, in practice, the value of $\varphi$ is unknown and, to satisfy the condition $v_{\rm max} \geq \varphi$, $v_{\rm max}$ should be selected sufficiently large. 

For the second statement, consider
\par\noindent\small
\begin{align}
\label{eq-positivstellensatz}
{\cal E}_i([\RV,v]^T) \!=\! \left\{\begin{array}{ll}
  1 & {\rm if} ~~ i=0, \\
  v {\cal J}(\RV) -{\cal F}(\RV), & {\rm if} ~~i=1, \\
   w_{i-1} (r_{i-1} - \lambda_{i-1}), & {\rm if} ~~ i=2, \cdots, M+1, \\
   r_{\Scale[0.5]{i-M-1}}, & {\rm if} ~~ i=M+2, \cdots, 2M+1, \\
  v_{\rm max} - v, & {\rm if} ~~ i=2M+2,
  \end{array}\right.
\end{align}
\normalsize
and that $\mathds{G}_p$ is defined according to \eqref{eq-positset}.
Construct ${\cal P}_i([\RV,v]^T)=0$ for $i=0,1,\cdots, 2M+1$ and ${\cal P}_{2M+2}([\RV,v]^T)=1$. It readily follows that $v_{\rm max} - v = \sum_{i=0}^{2M+2} {\cal P}_i([\RV,v]^T) {\cal E}_i([\RV,v]^T)$, thus $v_{\rm max} -v \in \mathds{G}_p$ with $p \geq 1$. Further, the set $\{v \in \mathds{R} \mid v_{\rm max} -v \geq 0\}$ is closed and bounded and, therefore, compact. This proves the second statement. 

The third statement requires establishing the strict positiveness of the objective on the feasible set of \eqref{eq-opt-compact}, i.e., $\mathds{W} = \{ \RV \in \mathds{R}^M,\; v \in \mathds{R} \mid \RV \succeq \ZEROVV,\; \WV \odot (\RV - \LAMBDA) \succeq \ZEROVV,\; v {\cal J}(\RV) -{\cal F}(\RV) \geq 0,\; v_{\rm max} - v \geq 0
\}$. Considering $a \in \mathds{R}_{ > 0}$ as a constant parameter independent of $\RV$ and $v$, it is always possible to replace $v$ with $v+a$ in the cost function of \eqref{eq-opt-compact} without affecting its solution. Then, it follows from \eqref{eq-16} that $v \geq {\cal L}(\RV)=\frac{{\cal F}(\RV)}{{\cal J}(\RV)} \geq 0$, thereby $v + a > 0$ on $\mathds{W}$ for any constant $a \in \mathds{R}_{ > 0}$.  This proves the third statement.

Consequently, according to Lemma \ref{lemma-1}, \eqref{eq:jsdp} is equivalent to minimization of  the positive function $v+a$ on the compact set $\mathds{W} = \{ \RV \in \mathds{R}^M,\; v \in \mathds{R} \mid {\cal E}_i([\RV,v]^T) \geq 0, \forall i \in \{1,2 ,\cdots, 2M+2\}
\}$ where ${\cal E}_i$'s are given in \eqref{eq-positivstellensatz}.
Now, resorting to \cite[Theorem 4.2]{lasserre2001global}, the resulting minimization problem can be equivalently recast as the SDP in \eqref{eq-opt-sdp}. This completes the proof.

\bibliographystyle{IEEEtran}
\bibliography{main}

\end{document}